\documentclass[aps,prl,10pt,twocolumn,showpacs, superscriptaddress]{revtex4-2}
\usepackage[utf8]{inputenc}
\usepackage{amsmath,amssymb,amsfonts,amsthm,amsopn,amscd,bm}
\usepackage{dsfont}
\usepackage{multirow}
\usepackage{graphicx}
\usepackage{booktabs}
\usepackage[shortlabels]{enumitem}
\usepackage{comment}
\usepackage{xcolor}
\usepackage{thmtools, thm-restate}

\usepackage[normalem]{ulem} 

\usepackage[breaklinks]{hyperref}
\hypersetup{colorlinks,
	linkcolor={red!40!black},
	citecolor={blue!50!black},
	urlcolor={blue!20!black}
}
\usepackage{orcidlink}

\makeatletter
\renewcommand*\env@matrix[1][*\c@MaxMatrixCols c]{%
	\hskip -\arraycolsep
	\let\@ifnextchar\new@ifnextchar
	\array{#1}}
\makeatother

\DeclareMathOperator{\tr}{tr}

\DeclareMathOperator{\rank}{rank}

\DeclareMathOperator{\one}{\mathds{1}}

\newcommand{\bra}[1]{\mathinner{\langle #1|}}
\newcommand{\ket}[1]{\mathinner{|#1\rangle}}
\newcommand{\braket}[2]{\mathinner{\langle #1|#2\rangle}}
\newcommand{\ketbra}[2]{\mathinner{|#1\rangle \langle#2|}}
\newcommand{\dyad}[1]{| #1\rangle \langle #1|}
\newcommand{\ot}[0]{\otimes}

\usepackage{accents}
\newcommand\thickbar[1]{\accentset{\rule{.4em}{.8pt}}{#1}}

\newcommand{\Vop}[0]{V}
\newcommand{\Wop}[0]{W}
\newcommand{\CZ}[0]{\operatorname{CZ}}
\newcommand{\CX}[0]{\operatorname{CX}}

\newcommand{\pX}[0]{X}
\newcommand{\pZ}[0]{Z}
\newcommand{\pY}[0]{Y}
\newcommand{\Had}[0]{H}
\newcommand{\CZbar}[0]{\operatorname{C\thickbar{Z}}}
\newcommand{\CXbar}[0]{\operatorname{C\thickbar{X}}}
\newcommand{\Xbar}[0]{\thickbar{X}}
\newcommand{\Zbar}[0]{\thickbar{Z}}
\newcommand{\bmR}[1]{\bm{\mathrm{#1}}}
\newcommand{\op}[1]{\operatorname{#1}}

\newcommand{\E}[0]{\operatorname{E}}

\newcommand{\inc}{\operatorname{in}}
\newcommand{\enc}{\operatorname{enc}}

\newcommand{\N}{\mathds{N}}

\newcommand{\C}{\mathds{C}}

\newcommand{\HH}{\mathcal{H}}
\newcommand{\PP}{{\partial \mathcal{B}}}
\newcommand{\MM}{\mathcal{B}}
\newcommand{\CC}{\mathcal{C}}

\newcommand{\regE}{E}
\newcommand{\regdE}{\partial E}
\newcommand{\regF}{F}
\newcommand{\regdF}{\partial F}

\usepackage{blindtext}

\makeatletter
\newcommand\emailx[1]{%
	\move@AF%
	\def\@affil{{\normalfont\,#1\strut}{}}%
}

\begin{document}
	
	\title{Engineering holography with stabilizer graph codes}
	
	\date{\today}
	
	\author{Gerard Anglès Munné${}^{\orcidlink{0000-0002-6168-9708}}$}
	\affiliation{Faculty of Physics, Astronomy and Applied Computer Science, 
		Institute of Theoretical Physics, Jagiellonian University, 
		30-348 Kraków, Poland}
	\emailx{ganglesmunne@gmail.com}
	
	\author{Valentin Kasper${}^{\orcidlink{0000-0001-7687-663X}}$}
	\affiliation{
		Institut de Ciències Fotòniques (ICFO), Mediterranean Technology Park, 08860 Castelldefels, Barcelona, Spain}
	\emailx{valentin.kasper@hotmail.de}
	
	\author{Felix Huber${}^{\orcidlink{0000-0002-3856-4018}}$}
	\affiliation{
		Institute of Theoretical Physics,
		Jagiellonian University, 
		30-348 Krak\'{o}w, 
		Poland}
	\affiliation{Bordeaux Computer Science Laboratory (LaBRI),
		University of Bordeaux,
		351 cours de la Liberation,
		33405,
		Talence,
		France}
	\emailx{felix.huber@physik.uni-siegen.de}

	\begin{abstract}
		The discovery of holographic codes 
		established a surprising connection
		between quantum error correction and the anti-de Sitter-conformal field theory correspondence.
		Recent technological progress in artificial quantum systems renders the experimental realization of such holographic codes now within reach. 
		Formulating the hyperbolic pentagon code in terms of a stabilizer graph code, 
		we give gate sequences that are tailored to systems with long-range interactions.
		We show how to obtain encoding and decoding circuits for the hyperbolic pentagon code, 
		before focusing on a small instance of the holographic code on twelve 
		qubits.
		Our approach allows to verify holographic properties by  partial decoding operations, recovering bulk degrees of freedom from their nearby boundary. 
	\end{abstract}
	
	\maketitle

	\section{Introduction}
	Holography emerged as a key concept in high-energy physics, gravity, and quantum information. With the introduction of the anti-de Sitter-conformal field theory (AdS-CFT) correspondence by Maldacena~\cite{Maldacena_1999}, holographic duality established a relation between two physical theories, one sitting in the bulk and the other sitting at the boundary of a hyperbolic space~\cite{Jahn_2021, Kohler_2019}. 
	In an effort to understand this AdS-CFT correspondence further, 
	geometrically arranged tensor networks arose as a useful tool, displaying unique entanglement properties~\cite{Hayden_2016,PhysRevLett.125.241602,PhysRevB.100.134203}
	This link between geometry and quantum entanglement led to 
	recent efforts seeking an experimental realization of holography~\cite{
		PhysRevLett.121.036403,
		PhysRevLett.128.013601,
		Kollar2019,
		Zhang2022}.
	However, an experimental realization of the proposed tensor networks is  challenging.

	The hyperbolic pentagon code proposed by~\cite{Pastawski_2015}, also known as HaPPY code, is today's premier toy model for understanding holographic duality. It is composed of a tensor-network with absolutely maximally entangled states (also known as perfect tensors) as basic building blocks. This model exhibits several desired features such as a uniform bulk and an entanglement entropy constrained by the Ryu-Takayanagi formula~\cite{Harlow_2017}.
	
	Despite recent theoretical investigations into the error-correcting capabilities of holographic codes
	\cite{PRXQuantum.3.020332,
		PRXQuantum.2.030337,
		PhysRevD.106.046009,
		PhysRevA.105.052446,
		PhysRevA.102.062417},
	experimental implementations have yet to come forward~\cite{Bhattacharyya_2022}. 
	In this work we close the gap towards an experimental implementation and bring
	the stabilizer approach to holography~\cite{PhysRevA.101.042305} to its logical conclusion by formulating it as a graph code.
	This opens up a path towards investigating AdS-CFT like models experimentally, 
	making their unique partial recovery features accessible for current and upcoming quantum technologies.

	Our framework yields a graph state from which we derive the gates necessary to encode, perform logical gates, and decode quantum information. 
	Additionally, it exhibits the essential characteristic of holographic systems, that is the ability to recover a bulk region from its nearby boundary. 
	
	While challenging, our toy model can already be implemented with as few as $12$ qubits, with experimental requirements that are within reach of current neutral atom platforms \cite{ebadi2021quantum}, superconducting qubits~\cite{arute2019quantum} and trapped ions experiments~\cite{zhang2017observation}.  
	
	Recent experimental efforts show the possibility to engineer long-range connectivity in neutral atom systems~\cite{Periwal2021, Bluvstein2022}, trapped ions~\cite{HAFFNER2008155,Monroe2021} and superconducting qubit platforms~\cite{arute2019quantum}. Depending on the platform the long-range interaction between the qubits is realized by coherently transporting the qubits or using a connecting bus.
	
	In contrast to more common methods which require stabilizer measurements to prepare the logical zero state~\cite{Abobeih_2022}, our approach reduces this task to a graph state preparation. For the trapped ions set up from~\cite{PRXQuantum.3.040310}, the graph state fidelity is estimated higher than the fidelity of the state prepared via stabilizer measurements.
	
	\section{Results}
	
	\subsection{Main contribution}
	
	We formulate the hyperbolic pentagon code introduced by Pastawski et al.~\cite{Pastawski_2015} in terms of its corresponding stabilizer graph code. 
	This allows to derive the encoding and decoding gates, as well as the partial recovery operations that are required to demonstrate holography experimentally.
	
	Our approach is based on three observations: 
	first, by choosing stabilizer states as building blocks, the hyperbolic pentagon code can be written in stabilizer form~\cite{PhysRevA.101.042305}.
	Second, a stabilizer code which encodes $k$ into $n$ qubits can be represented as a graph code~\cite{PhysRevA.65.012308}, that is, a graph state on $k+n$ systems. 
	Third, the number and range of interactions of this graph can be optimized through local Clifford operations, significantly reducing the experimental requirements to prepare the code states.
	
	This procedure allows us to represent the hyperbolic pentagon code as a stabilizer graph code in a manner suitable for experimental implementation
	(c.f. Fig.~\ref{fig:MinimalGraph}).
	In addition, our method also provides 
	the gates necessary to recover parts of the encoded bulk degrees of freedom from their nearby boundary, 
	thus demonstrating holographic features.
	
	Following the method described above, we propose an experimental implementation of a small instance of the hyperbolic pentagon code on twelve qubits 
	(c.f. Fig.~\ref{fig:MinimalGraph}).
	Through a stabilizer graph code representation, 
	we provide the gates required to encode and decode four bulk qubits into twelve boundary qubits. In addition, we show how to demonstrate holographic features of the logical code states. We list the specific gates needed for a partial decoding operation which recovers two bulk degrees of freedom from their nearby five-qubit boundary.
	
	This places the experimental implementation and certification of holographic systems within reach of current experimental
	high-connectivity platforms such as Rydberg atoms in optical tweezers, trapped ions, or cavity-coupled qubits. 
	Specifically by using trapped-ion setup from~\cite{PRXQuantum.3.040310}, we estimate a higher logical zero state fidelity with our method than with currently available ones~\cite{Abobeih_2022}.
	We suggest as first step towards the implementation of the holographic pentagon code, the preparation of the logical zero state.
	It is noteworthy that our proposal is of a scale that can be compared against numerical simulations.
	The reader only interested in the experimental implementation can directly jump to \hyperref[sect:Engholo]{\textit{A holographic model on 12 qubits}}.
	
	\begin{figure}[tbp]
		\centering
		\includegraphics[width=0.33\textwidth]{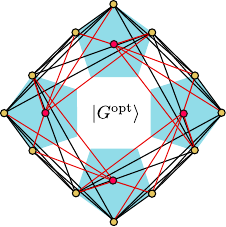}
		\caption{
			\label{fig:MinimalGraph}
			\textbf{A holographic graph code:}
			A small instance of the hyperbolic pentagon (HaPPY) code~\cite{Pastawski_2015},
			represented as a graph state or graph code.
			This representation is obtained from mapping
			the tensor network to a stabilizer state and then finding an experimentally suitable local Clifford equivalent graph state. 
			The entire construction can be described as a \hyperref[sect:holographcode]{\textit{holographic graph code}}.
		}
	\end{figure}
	
	\subsection{Holographic code} A holographic code encodes $k$ bulk qubits into $n$ boundary qubits with $n>k$, such that any bulk region can be recovered from its nearby boundary. The bulk degrees of freedom live in the Hilbert space $\HH_\MM$ and host the message, whereas the code space $\CC$ is a subspace of the boundary space $\HH_\PP$.
	The encoding of the bulk into the boundary is mathematically defined by a norm-preserving linear map $T_{\op{H}}:\HH_\MM \to \CC \subset \HH_\PP$, i.e. an isometry. The norm-preserving property of isometries can be written
	as $T^\dagger_{\op{H}} T_{\op{H}} = \one_{\MM}$. For our purposes the corresponding Hilbert spaces are
	$\HH_\MM={(\C^2)}^{\otimes k}$ and $\HH_\PP={(\C^2)}^{\otimes n}$ but qudit Hilbert spaces are also possible. Then, the isometry can explicitly be written as
	\begin{equation}\label{eq:isom}
		T_{\op{H}}=  \sum^{1}_{i_1, \dots, i_k=0} \ket{H_{i_1 \dots i_k}}_\PP\bra{i_1 \dots i_k}_\MM \,,
	\end{equation}
	mapping each computational basis element $\ket{i_1 \dots i_k}$ in $\HH_\MM$ to its corresponding logical state $\ket{H_{i_1\dots i_k}}$ in $\CC$. 
	Given the Pauli gates $\pX_j,\pY_j,\pZ_j$ acting on the $j$-th qubit in $\MM$, the isometry $T_{\op{H}}$ can be used to find the logical gates,  
	which act on $\CC$ in the same way as Pauli gates on $\HH_\MM$. 
	
	By turning the bra vector acting on $\MM$ into a ket, 
	the isometry in Eq.~\eqref{eq:isom} can be represented by an unnormalized quantum state 
	\begin{equation}\label{eq:state}
		\ket{H}=
		\sum^1_{i_1, \dots, i_k =0}  \ket{i_1 \dots i_k}_\MM \ot \ket{H_{i_1 \dots i_k}}_\PP\,,
	\end{equation}
	where we changed the order of the kets (bulk and boundary) for consistency with later sections.
	Therefore, the holographic code can be described by a state
	which we term {\em holographic state}.  
	From Eq.~\eqref{eq:state} we see that $\ket{H}$ is maximally entangled
	with respect to the  bipartition $\PP$ and $\MM$, which we denote by $\MM|\PP$. In reverse, each state of
	the form given in Eq.~\eqref{eq:state} induces an isometry. 
	Specifically, absolutely maximally entangled (AME) states, often referred to as perfect tensors when the number of parties is even, are maximally entangled with respect to any bipartition.
	
	An interesting way to construct a holographic code is via a tensor network that uses AME states as building blocks. Here we show how this holographic code can be understood as a graph code \cite{https://doi.org/10.48550/arxiv.1502.06618,Pastawski_2015}. 
	In a tensor network the tensors are connected by lines which 
	correspond to index contractions. 
	While contracting two AME states does not necessarily yield another AME state, the contraction is still an isometry.
	This fact can be directly seen from the tensor network representation
	of contracting two AME states.
	By assembling and contracting the tensors properly, one can construct a state that is maximally entangled across $\MM|\PP$. 
	Using such a maximally entangled state one can map bulk qubits to boundary qubits. 
	Since we use AME states as building blocks, the isometry is decomposed into further smaller isometries.
	These smaller isometries allow us to recover bulk qubits from their nearby boundary qubits, making the code holographic.
	
	The geometry of the tensor network determines how the decomposition of the 
	isometry is performed and, therefore, which information can be recovered. In this article we consider the hyperbolic pentagon code that contains six-qubit AME states as building blocks~\cite{Pastawski_2015}.
	This AME state was described first as {GF$(4)$}-{\em hexacode} in a seminal article by Calderbank et al.~\cite{681315} on the connection between classical and quantum stabilizer codes. The state was numerically rediscovered in Ref.~\cite{Borras_2007} and brought to graph state form in Ref.~\cite{https://doi.org/10.48550/arxiv.1306.2879}. When the number of parties is even, such states are often referred to as perfect tensors.
	Fig.~\ref{fig:holo}a illustrates the recovery for a specific boundary region: given the boundary qubits on the partition $\regdE$, it is possible to recover the bulk qubits on $\regE$ without using the qubits on~$\regdF$. 
	
	An important characteristic of holographic codes is
	that the Ryu-Takayanagi formula holds~\cite[Section 4]{Harlow_2017}.
	Roughly speaking, given a boundary bipartition $\regdE|\regdF$ and associated bulk regions $E|F$, 
	the Ryu-Takayanagi formula states that the entropy of the reduced state on $\regdE$ is proportional 
	to the length of the bulk geodesic 
	that separates $\regE$ and $\regF$, in addition to a bulk entropy term.
	For the holographic code considered here 
	the formula can be stated as
	$S(\regdE) \propto |\gamma|$ 
	for encoded product states,
	where $|\gamma|$ is the number of contracted indices in the tensor network that cross from the $\regE$ to $\regF$ (see Fig.~\ref{fig:holo}b). 
	Linked to this is  the ability to perform a partial bulk reconstruction, recovering a part of the bulk from its nearby boundary, a property which can be tested experimentally and will be addressed in \hyperref[sect:Pdecodingcircuit]{\textit{Partial decoding circuit}}.
	
	The construction of holographic codes via tensor networks
	is well suited to visualize the geometrical aspect of the
	code. On the other hand the stabilizer formalism is highly efficient
	in determining encoding and decoding strategies and to obtain the
	logical states and gates. Hence, we will discuss in the following section
	the stabilizer formalism in order to apply it to the hyperbolic pentagon code.
	
	\begin{figure*}[tbp]
		\centering
		\includegraphics[width=0.8\textwidth]{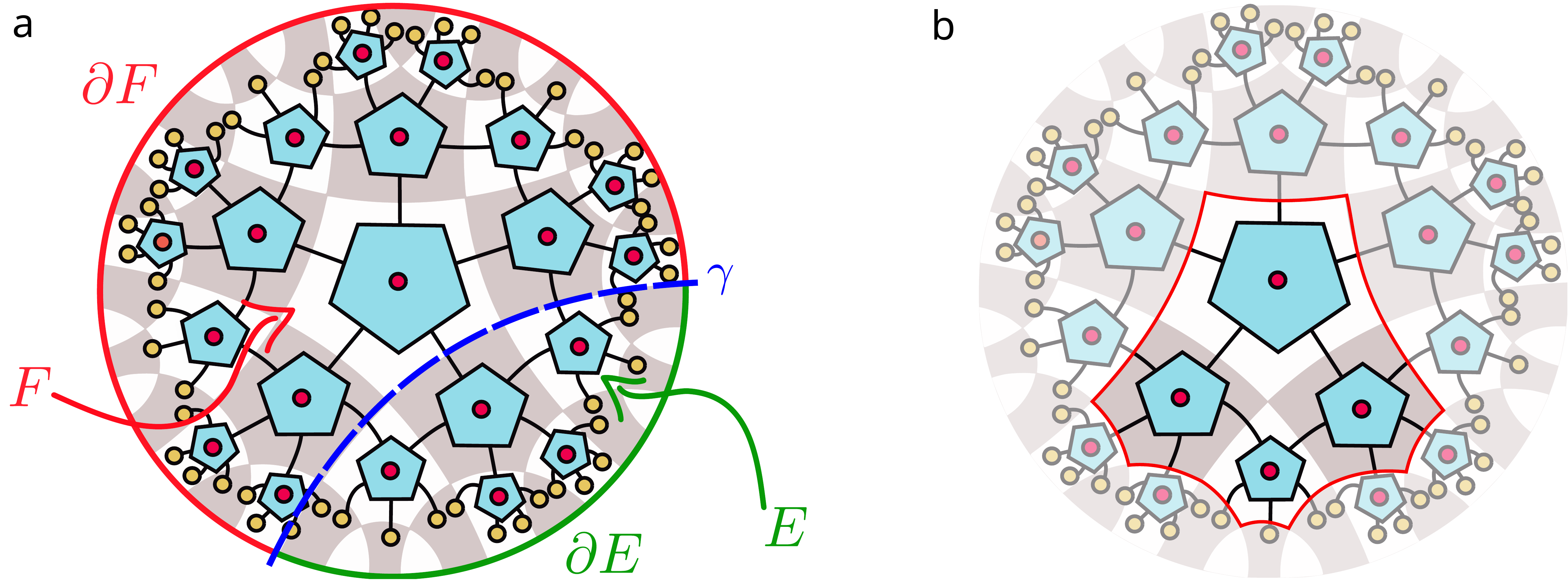}
		\caption{\textbf{Tensor network representation of the hyperbolic pentagon code.} The $k$ bulk qubits (in red) are encoded into $n$ boundary qubits (in gold). 
			Each pentagon represents a six-qubit perfect tensor, also known as absolutely maximally entangled (AME) state. 
			The geometry of the assembly makes the holographic state maximally entangled across $\MM|\PP$. 
			Fig.~\ref{fig:holo}a shows a specific region of the bulk $\regE$ which can be recovered by reading only the region $\regdE$ of the boundary, where
			the cut $\gamma$ is of size $|\gamma|=3$. 
			Fig.~\ref{fig:holo}b shows a small instance of the code (red) which we propose to prepare experimentally.}
		\label{fig:holo}
	\end{figure*}

	\subsection{Stabilizer states} An $m$-qubit stabilizer state is characterised by $\ell$ independent commuting operators $g_i$ that
	form the stabilizer $S=\langle g_1, \dots, g_l \rangle$, 
	with $\ell\leq m$ and $-\one \notin S$. 
	The $g_i$ are elements of the $m$-qubit Pauli group $\mathcal{P}^{m}$, which is formed by tensor products of Pauli matrices $\pX,\pY,\pZ$ and phases $\{\pm 1, \pm i\}$. Note however that the phases of the stabilizer elements are always real. 
	A stabilizer state \cite[Section 10.5.1]{nielsen_chuang_2010} can be written as
	\begin{equation}\label{eq:stab}
		\varrho
		=\frac{1}{2^\ell}\sum_{s \in S} s 
		=\frac{1}{2^\ell}\prod^\ell_{i=1}(\one+g_i) \,.
	\end{equation}
	The state $\varrho$ is proportional to a projector which acts on a subspace of dimension $2^{m-\ell}$. When $\ell=m$ the state is pure.
	
	A convenient way to represent a stabilizer state is through its {check matrix}. This is an $\ell \times 2m$ matrix $C = (A | B)$ whose rows correspond to the $\ell$ generators. 
	The matrix $A$ accounts for the $\pX$-part and $B$ for the $\pZ$-part of the generators: 
	$A_{ij} =1$ if $g_i$ contains an $\pX$ at position $j$,
	$B_{ij} =1$ if $g_i$ contains an $\pZ$ at position $j$, 
	$A_{ij} = B_{ij} =1$ if there is a $\pY$, 
	and $A_{ij} = B_{ij} =0$ if there is a $\one$. Therefore, the columns carry the information of how the generators act on individual qubits.
	The generators of two check matrices $C_1 = (A_1|B_1)$ and $C_2 = (A_2|B_2)$ commute if and only if
	\begin{equation}\label{eq:pc_comm}
		A_1 B^{\intercal}_2 - B_1 A^{\intercal}_2 = 0 \;.
	\end{equation}
	However, the parity check matrix does not carry all the information about $S$, since the signs of the generators are not included in $C$. 
	To keep track of the signs, we add an extra column $\omega$ to $C$
	such that 
	$C=(A|B|\omega)$,
	where 
	$\omega_i = 0$ if $g_i$ is positive 
	and 
	$\omega_i = 1$ if negative. 
	
	We recall that elementary row operations modulo $2$ on the check matrix leave the stabilizer invariant:
	the multiplication of generators corresponds 
	to the addition of the respective rows, and the multiplication and relabeling of generators does not affect $S$. It is important to remark that the multiplication of generators may change signs in $\omega$, e.g. $(\pX \otimes \pX)(\pZ \otimes \pZ)=-\pY \otimes \pY$. How the signs of the generators change is discussed in Appendix~\ref{app:stabilizer}.
	
	Graph states constitute a particular case of pure stabilizer states. These are defined by a graph of $m$ vertices connected by edges $e \in \E$. The generator associated with the vertex $i$ appearing in Eq.~\eqref{eq:stab} is
	\begin{equation}
		\label{eq:generators}
		g_i=\pX_i\bigotimes_{j\in N(i)}\pZ_j\,,
	\end{equation}
	where the neighbourhood $N(i)$ is the set of vertices $j$ connected to vertex $i$ by an edge.
	For a graph state, the check matrix reads $(\one | \Gamma)$ where $\Gamma$ is the adjacency matrix, representing the interaction between the qubits, and the phase vector is trivial $\omega=\bmR{0}$. An equivalent way to define a graph state is via controlled-$\pZ$ gates $\CZ_{uv}=\text{diag}(1,1,1,-1)$ acting on qubits $u$ and $v$ as
	\begin{equation}\label{eq:graphcz}
		\ket{G}=\prod_{(u,v)\in \E} \CZ_{uv} \ket{+}^{\otimes n}\,.
	\end{equation}
	
	Clifford operations are the unitaries that keep the Pauli group invariant under conjugation. An important example is the one-qubit Hadamard gate $\Had$ which acts as $\Had\pX\Had^\dag= \pZ$, $\Had\pZ\Had^\dag=\pX$ and $\Had\pY\Had^\dag=-\pY$. 
	It is known that all qubit stabilizer states are graph states up to 
	local Clifford operations (LC) \cite{PhysRevA.69.022316,10.3254/1-58603-660-2-115}.
	Many of the currently known AME states are graph states up to LC, with the notable exception of the recently discovered four-quhex AME state~\cite{PhysRevLett.128.080507}.

	\subsection{Holographic graph state}\label{sect:holographstate}
	Here we aim to find a suitable graph state~$\ket{G^{\text{opt}}}$ for experimental purposes which corresponds
	to the holographic state~$\ket{H}$ 
	and 
	to the tensor network in Fig.~\ref{fig:holo}a
	respectively.   
	
	An index contraction corresponds to projecting the tensor onto the Bell state and performing a partial trace. This fact can be seen from Eq.~\eqref{eq:Bellcontr} by defining an arbitrary state on $m$ qubits and a projector~$P^+=\ketbra{\phi^{+}}{\phi^{+}} \ot \one^{\ot m-2} $ with $\ket{\phi^{+}}=\sum^{1}_{r=0}\ket{rr}$ the (unnormalized) Bell state,
	\begin{align}
		\label{eq:Bellcontr}
		P^+\ket{\psi}
		&=
		\sum^1_{r,s=0}
		\sum^{1}_{i_1,\dots,i_{m}=0}
		\ket{rr} 
		\otimes 
		\braket{ss}{i_1i_2}
		\psi_{i_1\dots i_{m}}
		\ket{i_3\dots i_{m}}
		\nonumber\\
		&=\ket{\psi^+} \otimes \!\!\! \sum^{1}_{i_3,\dots,i_{m},s=0}\!\! 
		\psi_{ss i_3\dots i_{m}}
		\ket{i_3\dots i_{m}} 
		\nonumber\\ 
		&=\ket{\psi^+} \otimes \ket{\chi} \,.
	\end{align}
	Here, $\ket{\chi}$ is the state after the index contraction. 
	Note that the Bell state is a stabilizer state. Therefore, in case of $\ket{\psi}$ being a stabilizer state, so will $\ket{\chi}$ (see Ref.~\cite[Section 2]{Audenaert_2005}).
	That can be seen in more details in Appendix~\ref{app:stabilizer} where a method to obtain the generators of~$\ket{\chi}$ is found. 
	
	Our tensor network has the six-qubit AME state as a building block, which can be expressed as a stabilizer state~(see its graph representation in Fig.~\ref{fig:AME}). 
	Using this fact the contraction of the hyperbolic pentagon code leads to a stabilizer state $\ket{H}$.
	
	The state $\ket{H}$ can be transformed into a graph state $\ket{G}$ through a local Clifford operator $\Vop$ that is composed of a layer of one-qubit Hadamard gates followed by a layer of $\pZ$ gates,
	\begin{equation}\label{eq:Holotograph}
		\ket{G}=\Vop\ket{H}\,.
	\end{equation}
	This fact is shown in Appendix~\ref{app:graphH}
	The Hadamard gates transform the check matrix of $\ket{H}$ to $(\one|\Gamma|\omega)$, which is, up to the phase vector $\omega$, the check matrix of a graph state. The $\pZ$ gates applied set $\omega=\bmR{0}$. Note that, from the set of local Clifford operations, we only required Hadamard gates to transform $\ket{H}$ to a graph state, up to the signs of the generators.
	We emphasize that the contracted qubits 
	are not part of the holographic state but are only needed to construct the encoding.
	
	For a given graph state, there exist many other local unitary equivalent graphs. 
	While not all local unitary equivalent graph states 
	are local Clifford equivalent~\cite{DBLP:journals/qic/JiCWY10,Tsimakuridze_2017}, it significantly reduces the complexity of the problem by considering only the subset of local Clifford operations.
	To facilitate the implementation  
	where the boundary qubits are located according to their position in the tensor network, 
	we are interested in preparation protocols that require few interactions of shortest range only.
	In principle, a graph with such properties can be found by trying all possible mappings to a graph states, by applying local Clifford unitaries brute-force.

	A more refined strategy relies on the algorithm from Ref.~\cite{Adcock2020mappinggraphstate}: 
	This algorithm allows to generate the set of 
	LC-equivalent graph states on a small number of qubits. 
	The algorithm results in graphs that are non-isomorphic to each other. Thus, exploring all LC-orbit requires to additionally permute the associated vertices followed by checking whether the permuted graph is LC equivalent to the original graph. 
	In practice the following strategy appears useful: one explores the LC-orbit, chooses a graph with a few number of edges, and then further optimizes the range of interactions while making sure of LC equivalence to the original graph.
	We note however that these steps are computationally intensive: 
	The number of permutations grow super-exponentially. 
	Furthermore, while the 
	exact scaling of the LC algorithm is unknown, the LC orbit might also become super-exponentially large for $n\geq~12$~\cite{Adcock2020mappinggraphstate}. 
	
	\begin{figure}[tbp]
		\centering
		\includegraphics[width=0.35\textwidth]{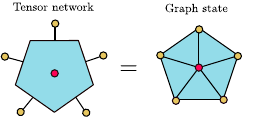}    
		\caption{\textbf{Building blocks of the tensor network.} 
			Each pentagon tensor with six indices (left) represents a six-qubit AME state, also known as perfect tensor. For our purposes we choose the graph state representation (right).}
		\label{fig:AME}
	\end{figure}
	
	For obtaining the graph $\ket{G^\text{opt}}$ (shown in Fig.~\ref{fig:MinimalGraph})
	we have limited ourselves to the following strategy:
	Exploring the LC-orbit for the graph in Fig.~\ref{fig:Minimal}b with $n=16$ 
	took our desktop computer a few seconds.  
	Then we permuted only the qubits which were part of the building blocks, and checked LC equivalence to the original graph.
	Here both the LC orbit and equivalence check were carried out with the library from Ref.~\cite{Adcock2020mappinggraphstate}. 
	Note that, while the number of its edges is minimal in the LC orbit, the interaction range of the resulting graph may be not.
	The LC orbit of a $n=22$ instance of the holographic code
	corresponding to $6$ pentagons can still be explored orbit in around $10$ minutes.
	
	However, larger instances seem to require a more heuristic strategy:
	Instead of exploring the LC-orbit, we then apply heuristically Hadamard gates to $\ket{H}$ that result in graph states [see Eq.~\eqref{eq:Holotograph}]. 
	Finally, $\pZ$ gates can always be applied to set $\omega=\textbf{0}$. Thus they do not play an important role in the optimization.
	This last strategy seems to work better for symmetric instances of the hyperbolic pentagon code. Surprisingly, it also leads to the improved graph from Fig.~\ref{fig:MinimalGraph}. 
	Appendix~\ref{app:larger} shows a larger instance resulting from contracting eleven AME states corresponding to $n=36$ qubits.
	
	\subsection{Holographic graph code}\label{sect:holographcode}
	Here we show how the holographic code can be understood as a graph code \cite{PhysRevA.65.012308,10.1007/978-3-642-20901-7_9}. 
	From the graph code and its representation as 
	graph state $\ket{G}$ [see Eq.~\eqref{eq:Holotograph}], 
	we derive the logical basis states [see Eq.~\eqref{eq:state}]. 
	We emphasize that the results in this section applies for any graph code, including $\ket{G^\text{opt}}$. 
	
	We recall that the check matrix of a graph state is described by
	$(\one | \Gamma)$. 
	Since the holographic graph state $\ket{G}$ is maximally entangled across $\MM|\PP$, we can write its adjacency matrix as
	\begin{equation}\label{eq:check0}
		\begin{pmatrix}[cc|cc]
			\one_n & 0 & \Gamma_\PP & B^\intercal\\
			0 & \one_k & B & \Gamma_\MM
		\end{pmatrix}\,,
	\end{equation}
	where $\rank(B)=k$, as shown in \cite[Section II]{Grassl}. The $n$ first columns of the $\pX$- and $\pZ$-part carry the information how the generators act on the boundary qubits and the remaining $k$ columns describe how the generators act on the logical qubits. Here $\Gamma_\MM$, $\Gamma_\PP$ and $B$ 
	represent the interactions within the bulk, 
	within the boundary, and between bulk and boundary qubits respectively.
	The corresponding edge sets are $E_\MM$, $E_\PP$, $E_{\PP|\MM}$.
	
	Similarly to $\ket{H}$ from Eq.~\eqref{eq:state}, the graph state $\ket{G}$ from Eq.~\eqref{eq:check0} can be written as
	\begin{equation}\label{eq:graphst}
		\ket{G} = \sum^1_{i_1,\dots,i_k =0}
		\ket{i_1 \dots i_k}_\MM \ot
		\ket{G_{i_1 \dots i_k}}_\PP 
		\,.
	\end{equation}
	The basis elements of the code space $\ket{G_{i_1 \dots i_k}}$ with $i_1,\dots,i_k\in \{0,1\}$ can be defined as
	\begin{equation}\label{eq:logicalbasis}
		\ket{G_{i_1 \dots i_k}}= (-1)^{\sum_{(u,v)\in \E_\MM} i_vi_u}
		\prod^{k}_{r=1} \Xbar^{i_r}_r \ket{G_{0 \dots 0}} \,,
	\end{equation}
	where
	$\ket{G_{0\dots0}}$  
	is the logical zero state and $\{\Xbar_r\}^k_{r=1}$ are the logical $\Xbar$~gates,
	\begin{equation}\label{eq:logical0X}
		\ket{G_{0 \dots 0}}  = \!\!\prod_{(u,v) \in \E_\PP} \!\! \CZ_{uv} \ket{+}^{\otimes n} \,,\,\,\, \Xbar_r =\bigotimes^{n}_{s=1} \pZ^{B_{rs}}\,.
	\end{equation}
	Eq.~\eqref{eq:logicalbasis} and \eqref{eq:logical0X} are derived in Appendix~\ref{app:logIso}. 
	
	Now, we want to find the remaining logical gates $\{\Zbar_{r}\}^k_{r=1}$  and the generators of the code space $\{g_r\}^{n-k}_{r=1}$. It is clear that the generators must commute with the logical gates since their action leaves the code space invariant.
	As they act on~$\PP$ only, we consider the boundary qubits of the check matrix Eq.~\eqref{eq:check0} and write it as~\cite[Section 4]{10.1007/978-3-642-20901-7_9}
	
	\begin{equation}\label{eq:checkphy}
		\begin{pmatrix}[c|c]
			\one_n & \Gamma_\PP \\
			0 & B 
		\end{pmatrix}
		=
		\begin{pmatrix}[cc|cc]
			\one_{n-k} & 0  & \Gamma_1 & \Gamma^\intercal_2  \\
			0 & \one_k & \Gamma_2 & \Gamma_3 \\
			0 & 0 & B_1 & B_2 
		\end{pmatrix}\,,
	\end{equation}
	such that $\rank(B_2)=k$. Here $B$ is decomposed into two matrices $B_1,B_2$ and $\Gamma_\PP$ into three matrices $\Gamma_1, \Gamma_2, \Gamma_3$. 
	
	Performing row operations on Eq.~\eqref{eq:checkphy} leads to the generators of the code subspace $\CC_G$ and the logical operators $\thickbar{Z}_{}$ and
	$\thickbar{X}_{}$. The result is given by
	\begin{widetext}
		\begin{equation}\label{eq:checkLogOpxz}
			\begin{pmatrix}[c]
				\CC_{G} \\ 
				\thickbar{Z}_{}  \\
				\thickbar{X}_{} 
			\end{pmatrix}=
			\begin{pmatrix}[cc|cc]
				\one_{n-k} & -B^\intercal_1(B_2^\intercal)^{-1}& [\Gamma_{1}-B^\intercal_1(B_2^\intercal)^{-1}\Gamma_2] &  [\Gamma^\intercal_{2}-B^\intercal_1(B_2^\intercal)^{-1}\Gamma_3] \\
				0 & (B^\intercal_2)^{-1}& (B^\intercal_2)^{-1}\Gamma_2 & (B^\intercal_2)^{-1}\Gamma_3 \\ 
				0 & 0 & B_1 & B_2 \\ 
			\end{pmatrix}\,.
		\end{equation}
	\end{widetext}
	The proof of this result can be found in Appendix~\ref{app:OpCode}.
	
	It is important to remark that with the row operations needed to obtain Eq.~\eqref{eq:checkLogOpxz}, 
	the phase vector can also change. This will also introduce additional signs (see Appendix~\ref{app:OpCode} for details) to the logical $\Zbar$~gates and the generators of the code space associated to $\ket{G}$. 
	
	\subsection{Encoding} 
	We now describe how an arbitrary state is encoded into the holographic code. 
	The layout of Fig.~\ref{fig:holo} contains
	$k$ bulk and $n$ boundary qubits.
	Given the logical gates and the generators for the code subspace [c.f. Eq.~\eqref{eq:checkLogOpxz}], we provide a recipe to encode an arbitrary $k$-qubit bulk state into the boundary.
	This recipe is a modification of the encoding method from Ref.~\cite{10.1007/978-3-642-20901-7_9} and
	can be found in more details in Appendix~\ref{app:Encoding}.
	
	Define the controlled logical gates acting on bulk qubit $j$ in $\MM$ as 
	\begin{equation}
		\begin{aligned}
			\CXbar_j = &
			\ketbra{0}{0}_{j} \otimes \one +\ketbra{1}{1}_{j}  \otimes \Xbar \,,\\
			\CZbar_j = &
			\ketbra{0}{0}_{j} \otimes \one +\ketbra{1}{1}_{j}  \otimes \Zbar
			\,,
		\end{aligned}
	\end{equation}
	with $\one$, $\Xbar$, and $\Zbar$ acting on the boundary space $\PP$.

	A bulk state $\ket{\phi_{\inc}}$ encodes into a boundary state $\ket{\phi_{\enc}}$ via
	\begin{equation}
		\label{eq:gateUG}
		\ket{+}^{\otimes k}_{\MM} \ot \ket{\phi_{\enc}}_{\PP} =
		U_G 
		\big(
		\ket{\phi_{\inc}}_{\MM} \ot \ket{+}_{\PP}^{\otimes n}  
		\big)\,,
	\end{equation}
	where $U_G = U_3  U_2  U_1$ with
	\begin{equation}\label{eq:gates}
		\begin{aligned}
			U_1  & = \!\!\!\!\!\!\!\! \prod_{
				\substack{
					(u,v)  \in \{\E_\MM, \E_\PP\} 
				} 
			} \!\!\!\!\!\! \CZ_{uv}\,, \quad U_2 =\prod^{k}_{j=1} \CXbar_{j}\,, \\
			U_3 &= \prod^{k}_{j=1} {\CZbar}_{j} 
			\Big( 
			\Had^{\otimes k}_{\MM} \ot \one^{\ot n}_{\PP} 
			\Big) \,.
		\end{aligned}
	\end{equation}
	After a successful encoding
	the bulk degrees of freedom are left in the product state $\ket{+}^{\ot k}$, as illustrated in Fig.~\ref{fig:enc}. 
	
	The unitary $U_G$ is decomposed into three unitaries.
	The gates in $U_1$ take into account the interactions among boundary qubits and bulk qubits separately [$\Gamma_\PP$ and $\Gamma_\MM$ from Eq.~\eqref{eq:check0}]. They prepare the logical zero state and the phases of the logical states respectively. 
	The gates in $U_2$ take into account the interactions between bulk and boundary qubits [$B$ from Eq.~\eqref{eq:check0}], entangling the bulk computational basis with the logical boundary states. 
	Finally, $U_3$ is responsible for disentangling both systems with the information from the bulk transmitted to the boundary.
	Note that the conditional gates $\CXbar_j$ and $\CZbar_j$ act between the qubit $j$ in $\MM$ and the boundary $\PP$, 
	a gate $\CZ_{uv}$ acts on two qubits $u,v$ that are both in either $\MM$ or $\PP$ only.
	
	Since the logical gates are composed by a tensor product of Pauli gates, we can decompose $\CXbar_j$ and $\CZbar_j$ into a product of $\text{CX}_{uv}$ and $\CZ_{uv}$ gates. 
	We can see that with an example. Define a controlled gate composed by Pauli gates:
	\begin{equation}
		\ketbra{0}{0}_1\otimes \one^{\otimes 4}_{2345} + \ketbra{1}{1}_1 \otimes \big(\pX_2 \pZ_3\pY_4 \pY_5\big)\,.
	\end{equation}
	One can decompose it as
	\begin{equation}
		\CZ_{14}\big(\text{CX}_{12}\text{CX}_{14}\text{CX}_{15}\big)\big(\CZ_{13}\CZ_{15}\big)\,.
	\end{equation}
	This decomposition is useful since $\text{CX}$ and $\CZ$ gates are realizable in many experimental platforms.

	\begin{figure}[tbp]
		\centering
		\includegraphics[width=0.45\textwidth]{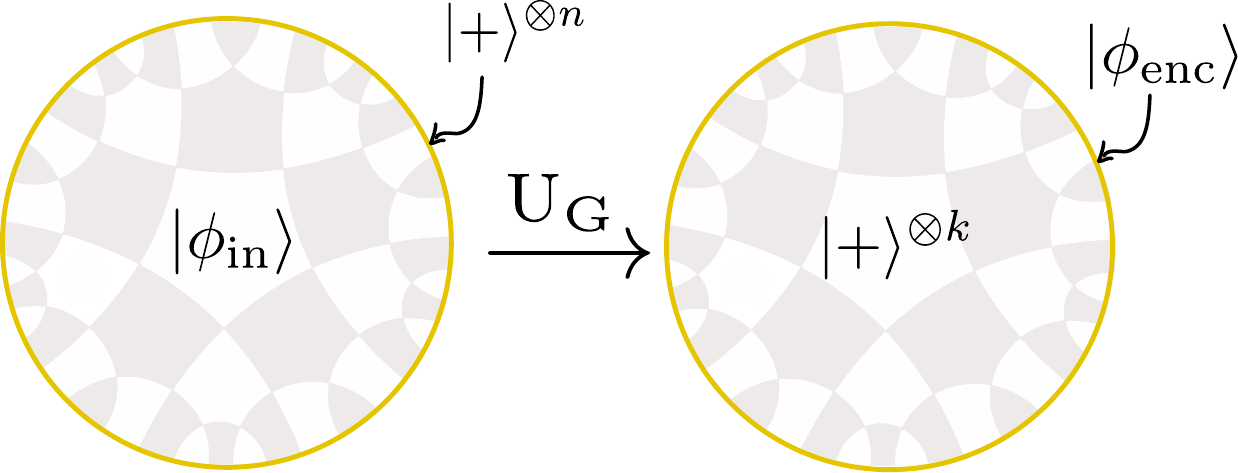} \caption{
			\label{fig:enc}
			\textbf{Encoding scheme.} 
			The state $\ket{\phi_{\inc}}$ to be encoded is localized in the bulk, while the boundary qubits are in the product state $\ket{+}^{\otimes k}$. 
			After performing the encoding unitary $\text{U}_G$ the information
			is mapped to the boundary state $\ket{\phi_{enc}}$ and with the bulk  
			in a product state $\ket{+}^{\otimes k}$.
		}
	\end{figure}

	\subsection{Partial decoding}
	The decoding abilities of the hyperbolic pentagon code are related to its geometry, which is induced by how the AME states are arranged in the tensor network. 
	Their contraction yields the holographic state $\ket{H}$, 
	which represents an isometry $T_{\op{H}}$ that encodes $k$ bulk qubits into $n$ boundary qubits through Eq.~\eqref{eq:isom}.
	
	To see how 
	a bulk region 
	can be recovered from its nearby boundary,
	we partition the bulk into two complementary regions $\regE$ and $\regF$ along a cut $\gamma$ with associated boundary regions $\regdE$ and $\regdF$, shown in Fig.~\ref{fig:holo}a.
	To each tensor network leg crossed by the cut $\gamma$ we associate a qubit. These qubits form the Hilbert space $\mathcal{H}_{\gamma}=(\mathbb{C}^2)^{\otimes |\gamma|}$.
	If there exists an isometry from $\regE \cup \gamma$ to $\regdE$, then the quantum information stored in $\regE$ can be recovered from $\regdE$. 
	A way to check whether such isometry exists is through the method of tensor pushing~\cite[Section 5.3]{Pastawski_2015}. Such a partial isometry can be written as,
	\begin{align}
		T_{\op{h}}
		&=
		\sum_{\bmR{i}\in \mathbb{Z}^{|\regE|}_2}\sum_{\bmR{j}\in \mathbb{Z}^{|\gamma|}_2} 
		\ket{h_{\bmR{ij}}}_{\regdE}
		\bra{\bmR{i}}_E 
		\ot 
		\bra{\bmR{j}}_\gamma \nonumber\\
		&=\sum_{\bmR{\ell}\in \mathbb{Z}^{|\regdE|}_2}
		\sum_{\bmR{i}\in \mathbb{Z}^{|\regE|}_2}
		\sum_{\bmR{j}\in \mathbb{Z}^{|\gamma|}_2}  
		h_{\bmR{\ell ij}}\ket{\bmR{\ell}}_{\regdE}
		\bra{\bmR{i}}_E 
		\ot 
		\bra{\bmR{j}}_\gamma\,,
	\end{align}
	where
	$\ket{h_{\bmR{ij}}}_{\regdE}$, 
	forms an orthonormal basis of $\regdE$.
	
	Given a bulk state that was encoded through the isometry $T_{\op{H}}$, one can apply then a partial isometry~$T^{\dagger}_{\op{h}}$ on $\regdE$ to recover the bulk region~$\regE$. 
	This follows from the fact that $T_{\op{H}}$ can be decomposed in terms of the elements $h_{\bmR{\ell ij}}$ from $T_{\op{h}}$ as,
	\begin{equation}
		T_{\op{H}}
		=
		\sum_{\bmR{\ell}\in \mathbb{Z}^{|\regdE|}_2}
		\sum_{\bmR{i}\in \mathbb{Z}^{|\regE|}_2}
		\sum_{\bmR{j}\in \mathbb{Z}^{|\gamma|}_2}
		\Big(
		h_{\bmR{\ell ij}} 
		\ket{\bmR{\ell}}_{\regdE}\bra{\bmR{i}}_{\regE} \ot
		{R_{\bmR{j}}}
		\Big)\,,
	\end{equation}
	where $R_{\bmR{j}}$ is a tensor mapping from  $\regF$ to $\regdF$. 
	This decomposition emerges from the tensor network illustrated in Fig.~\ref{fig:holo}a.
	The isometry $T_{\op{h}}$ is constructed by the building blocks geometrically situated in $E$,  
	which in turn are contracted by $|\gamma|$ indices to the rest of the building blocks situated in $F$.
	
	One sees that 
	$T_{\op{H}}$ followed by $T^{\dagger}_{\op{h}}$ acts as identity on $\regE$,
	\begin{align}\label{eq:partialiso}
		&\phantom{=}(T^\dagger_h \ot \one_{\regdF})T_{\op{H}}  \nonumber \\
		&=
		\sum_{\bmR{i},\bmR{i'}\in \mathbb{Z}^{|\regE|}_2}
		\sum_{\bmR{j},\bmR{j'}\in \mathbb{Z}^{|\gamma|}_2}
		\Big(\sum_{\bmR{\ell}\in \mathbb{Z}^{|\regdE|}_2}
		h^*_{\bmR{\ell i'j'}} h_{\bmR{\ell ij}}\Big) 
		\ketbra{\bmR{i'}}{\bmR{i}} \ot \ket{\bmR{j'}}_\gamma
		\ot
		{R_{\bmR{j}}} \nonumber \\ 
		&= \one_{\regE} \ot \sum_{\bmR{j}\in \mathbb{Z}^{|\gamma|}_2}\ket{\bmR{j}}_\gamma\ot {R_{\bmR{j}}}
		\,,
	\end{align}
	where we used that $\sum_{\bmR{\ell}\in \mathbb{Z}^{|\regdE|}_2} h^*_{\bmR{\ell i'j'}} h_{\bmR{\ell ij}}=\delta_{\bmR{ii'}} \delta_{\bmR{jj'}}$ because the elements $\ket{h_{\bmR{ij}}}$ form an orthonormal basis. 
	Hence, $T_{\op{h}}$ recovers the bulk information of $\regE$ from its nearby boundary $\regdE$.

	The isometry $T_{\op{h}}$ can be represented as a quantum state $\ket{h} \in \HH_{\regE} \otimes \HH_{\gamma} \otimes \HH_{\regdE}$. This state is constructed by contracting six-qubit AME states from the regions $\regE$ and $\regdE$ only and it can be converted to a graph state $\ket{g}$ via local Clifford operations $W$, such that $\ket{g}=W\ket{h}$. 
	Eq.~\eqref{eq:gateUG} and \eqref{eq:gates}
	allow to obtain
	the encoding and decoding gates $U_{\op{g}}$ corresponding to $\ket{g}$.
	As done with the full code~$\ket{G}$, it is practical to optimize this "partial" graph code~$\ket{g}$ with respect to the range and number of gates.
	
	Recall that $\ket{G}=\Vop\ket{H}$ and $\ket{g}=\Wop\ket{h}$, 
	where $\Wop,\Vop$ are composed of local Clifford gates. Then, the decoding gate $U_{\op{g}}$ has to be corrected with corresponding local Clifford gates (see Appendix~\ref{app:Pdec}), 
	\begin{equation}\label{eq:recgates2}
		{\widetilde{U}^\dagger_{\op{h}}}
		=
		(\Vop^\dagger_{\regE}\Wop_{\regE} \ot \one_{\gamma\regdE})
		\,
		U^\dagger_{\op{g}}
		\, 
		(\one_{\regE\gamma}\ot \Wop_{\regdE}\Vop^\dagger_{\regdE})\,,
	\end{equation}
	where $
	{V_\regE},
	{V_{\regdE}},
	{W_\regE},
	{W_{\regdE}}$ 
	contain the local Clifford gates of $V$ and $W$ having support on $\regE$ and $\regdE$ respectively.
	Here, ${\tilde{U}_{\op{h}}}$ is the unitary operator that performs a partial decoding
	of a boundary state that was encoded via $\ket{G}$.
	
	Recall that the bulk decomposes as 
	$\MM = \regE \cup \regF$ 
	and the boundary as
	$\PP = \regdE \cup \regdF$.
	Then a state $\ket{\phi_{\text{enc}}}$ that was encoded through $\ket{G}$,
	\begin{equation}
		\ket{+}^{\otimes k}_{\MM} \ot \ket{\phi_{\enc}}_{\PP} =
		U_{\op{G}} 
		\big(
		\ket{\phi_{\inc}}_{\MM} \ot \ket{+}_{\PP}^{\otimes n}  
		\big)\,,
	\end{equation}
	can be partially decoded by ${\tilde{U}^{\dagger}_{\op{h}}}$ 
	\begin{align}
		&\ket{\psi_\text{dec}} =
		\big(
		{\tilde{U}^{\dagger}_{\op{h}}}
		\otimes 
		\one_{\regdF} 
		\big) 
		\big(
		\ket{+}^{\otimes (
			|\regE|+|\gamma|)}_{\regE\gamma} \otimes \ket{\phi_{\enc}}_{\PP}
		\big)
		\,.
	\end{align}
	In particular, one can check that for this partially decoded state
	\begin{equation}\label{eq:partialdecoding}
		\tr_{\regF}  \big(\ketbra{\phi_{\inc}}{\phi_{\inc}}\big)
		=
		\tr_{\gamma\PP}  \big(\ketbra{\psi_{\text{dec}}}{\psi_{\text{dec}}}\big)
	\end{equation}
	holds. Consequently, all quantum information contained in $\regE$ can be recovered from its nearby boundary $\regdE$, demonstrating holographic properties. Which other recovery regions are possible is studied in Ref.~\cite[Section 5.3]{Pastawski_2015}. 
	
	\begin{figure}[tbp]
		\centering
		\includegraphics[width=0.33\textwidth]
		{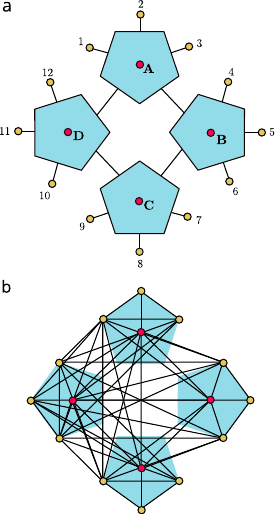}
		\caption{\textbf{Four pentagons.} 
			Contracting four perfect tensors (upper) leads to a small instance of the holographic code. 
			The contraction is mapped to a graph state (lower).
			This state can be thought of as an isometry
			that encodes the bulk qubits (red) 
			into the boundary qubits (golden).
			Optimizing the range and the number of edges of the graph in Fig.~\ref{fig:holo}b over local Clifford operations leads to the graph in Fig.~\ref{fig:MinimalGraph}.
			\label{fig:Minimal}
		}
	\end{figure}

	\begin{figure}[tbp]
		\centering
		\includegraphics[width=0.48\textwidth]{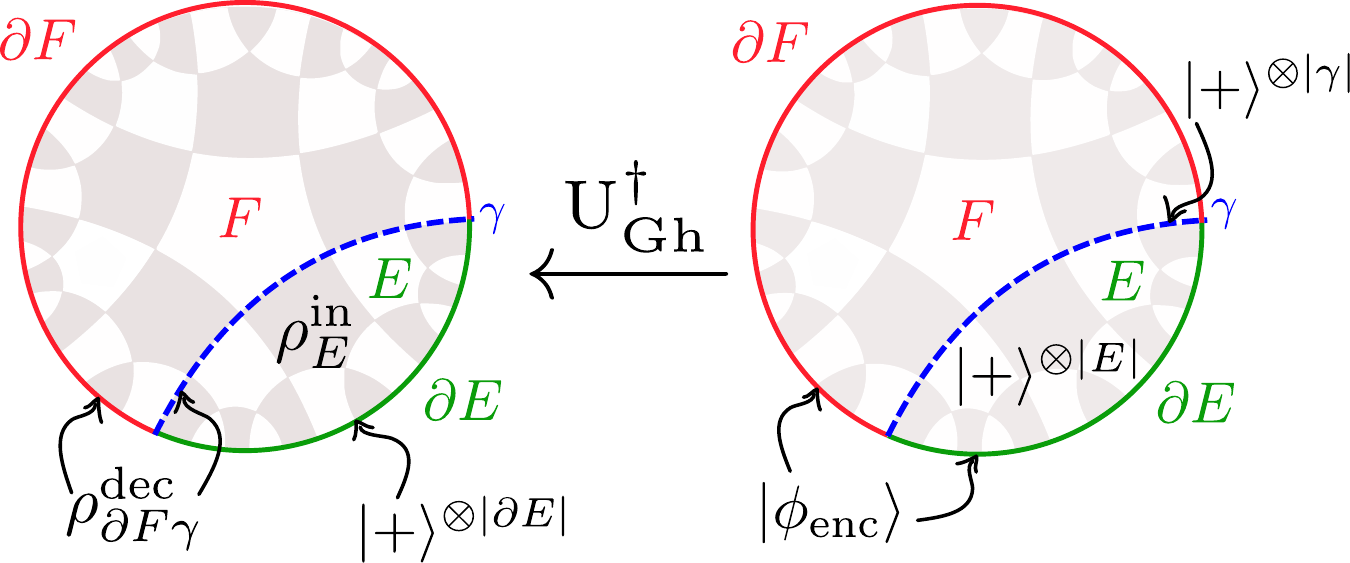}
		\caption{\textbf{Partial decoding scheme.}  
			The encoded state $\ket{\phi_{\enc}}$ is localized on the boundary, and the bulk qubits on $\regE$ and the qubits on the cut $\gamma$ are in the product states $\ket{+}^{\ot|E|}$ and $\ket{+}^{\ot |\gamma|}$ respectively.
			The partial decoding unitary ${\widetilde{U}^\dagger_{\op{h}}}$ transmits the information from the boundary region $\regdE$ to $\regE$. 
		}
		\label{fig:Dec_1}
	\end{figure}
	
	\subsection{A holographic model on 12 qubits} \label{sect:Engholo}
	
	Here we consider a small instance of the holographic code on $12$ qubits that exhibits holographic properties and
	describe the experimental preparation of its logical states, 
	its encoding, as well as the decoding procedures. 
	We also show how to recover a bulk region 
	from its nearby boundary.
	The methodology relies on the formulation of the hyperbolic pentagon (HaPPY) code~\cite{Pastawski_2015} as a stabilizer graph code derived in the previous sections.
	
	This toy model consists of four connected pentagons, 
	each representing a six-qubit absolutely maximally entangled state (AME), also known as perfect tensor, as illustrated in Fig.~\ref{fig:Minimal}a. 
	The toy model involves twelve boundary qubits, labeled by $1$ to $12$, and four bulk degrees of freedom labeled by $\text{A}, \text{B}, \text{C}, \text{D}$,
	thus requiring $16$ qubits in total.

	The building block of the hyperbolic pentagon code is the 
	six-qubit AME state, for which a highly symmetric graph state representation exists (c.f. Fig.~\ref{fig:AME})  \cite{https://doi.org/10.48550/arxiv.1306.2879}.
	The contraction of four such states yields the stabilizer state~$\ket{H}$ which carries both information about the code subspace as well as about the encoding.
	The resulting $\ket{H}$ can be transformed to a graph state~$\ket{G}$ by the application of a single layer of Hadamard gates and a subsequent layer of $\pZ$ gates, as shown in Appendix~\ref{app:graphH}.
	As $\ket{G}$ requires many long-ranges $\CZ$ gates,
	it is useful to choose a local Clifford equivalent graph state
	that requires gates 
	of shortest possible range. 
	Here we choose the state $\ket{G^\text{opt}}$ that is shown in Fig.~\ref{fig:MinimalGraph}, which can be found by exploring the local Clifford orbit.
	This graph has no edges that cross the center and is rotational invariant. 
	
	From the graph representation of the state $\ket{G^{\text{opt}}}$ as illustrated in Fig.~\ref{fig:MinimalGraph}, 
	the logical zero state and the logical $\Xbar$~gates are extracted as follows: 
	remove the red vertices (the bulk qubits) and their incident edges one  obtain
	the logical state $\ket{G^{\text{opt}}_{0000}}$,
	corresponding to Eq.~\eqref{eq:logical0X}. This logical state
	is illustrated in Fig.~\ref{fig:LogicalState}. 
	The logical $\Xbar$~gates,
	given through Eq.~\eqref{eq:logical0X},
	are determined by the red
	edges $B_{rs}$ that connect the bulk with the boundary
	and read
	\begin{equation}\label{eq:LogicalX}
		\begin{matrix}[cc ccc ccc ccc ccc]
			&&\text{\footnotesize1}&\text{\footnotesize2}&\text{\footnotesize3}&\;
			\text{\footnotesize4}&\text{\footnotesize5}&\text{\footnotesize6}&\;
			\text{\footnotesize7}&\text{\footnotesize8}&\text{\footnotesize9}&\; \text{\footnotesize10}&\!\text{\footnotesize11}&\!\!\!\!\text{\footnotesize12}\\
			\thickbar{X}_{\op{A}} &=& \textcolor{red}{\pZ}&\one& \textcolor{red}{\pZ}&\,
			\textcolor{red}{\pZ}& \textcolor{red}{\pZ}& \textcolor{red}{\pZ}&\, \one&\one&\one&\,
			\one\!&\one&\!\one \,, \\
			\thickbar{X}_{\op{B}} &=&
			\one&\one&\one&\,
			\textcolor{red}{\pZ}&\one&
			\textcolor{red}{\pZ}&\,
			\textcolor{red}{\pZ}& \textcolor{red}{\pZ}& \textcolor{red}{\pZ}&\,
			\one\!&\one&\!\one \,, \\
			{{\thickbar{X}_{\op{C}}}} &=&
			\one&\one&\one&\,
			\one&\one&\one&\,
			\textcolor{red}{\pZ}&\one&
			\textcolor{red}{\pZ}&\,
			\textcolor{red}{\pZ}& \textcolor{red}{\pZ}& \textcolor{red}{\pZ} \,, \\
			\thickbar{X}_{\op{D}} &=&
			\textcolor{red}{\pZ}\!&\textcolor{red}{\pZ}& \!\textcolor{red}{\pZ}&\,
			\one&\one&\one&\,
			\one&\one&\one&\,
			\textcolor{red}{\pZ}\!&\one&
			\!\textcolor{red}{\pZ} \,.
		\end{matrix}
	\end{equation}
	
	While the logical $\Xbar$~gates can be extracted directly from the graph in Fig.~\ref{fig:MinimalGraph}, the logical $\Zbar$ operator do not seem to have such simple graphical interpretation. However, they can be found in Eq.~\eqref{eq:checkLogOpxz} and read 

	\begin{equation}\label{eq:LogicalZ}
		\begin{matrix}[cc ccc ccc ccc ccc]
			&&\text{\footnotesize1}&\text{\footnotesize2}&\text{\footnotesize3}&\;
			\text{\footnotesize4}&\text{\footnotesize5}&\text{\footnotesize6}&\;
			\text{\footnotesize7}&\text{\footnotesize8}&\text{\footnotesize9}&\; \text{\footnotesize10}&\!\text{\footnotesize11}&\!\!\!\!\text{\footnotesize12}\\
			\thickbar{Z}_{\op{A}}&=&
			\one&\pZ&\pZ&\;
			\pZ&\pX&\one&\;
			\pZ&\pZ&\pZ&\;
			\one\!&\one&\!\one \,,\\
			\thickbar{Z}_{\op{B}}&=&
			\one&\one&\one&\;
			\one&\pZ&\pZ&\;
			\pZ&\pX&\one&\, 
			\pZ\!&\pZ&\!\pZ\,,\\
			\thickbar{Z}_{\op{C}}&=&
			\pZ&\pZ&\pZ&\;
			\one&\one&\one&\;
			\one\!&\pZ&\!\pZ&\;
			\pZ&\pX&\one\,,\\ 
			\thickbar{Z}_{\op{D}}&=&
			\pZ&\pX&\one&\;
			\pZ&\pZ&\pZ&\;
			\one&\one&\one&\,
			\one\!&\pZ&\!\pZ\,.
		\end{matrix}
	\end{equation}
	Note that the logical gates in Eq.~\eqref{eq:LogicalX} and Eq.~\eqref{eq:LogicalZ} preserve the same rotational symmetry as the graph (Fig.~\ref{fig:MinimalGraph}) from which they are extracted. 
	
	The logical operators can be further simplified when multiplied by the code subspace generators. However, we lose the graphical interpretation of the logical $\Xbar$ gates.
	Similarly, the code subspace generators can be optimized by multiplying each other. This can reduce the number of gates used in this section even more (see Appendix~\ref{app:fidelity}).
	
	\begin{figure}[tbp]
		\centering
		\includegraphics[width=0.36\textwidth]{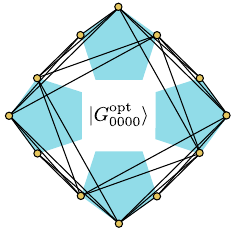}
		\caption{\label{fig:LogicalState}\textbf{Logical state.} 
			The logical state $\ket{G^{\text{opt}}_{0000}}$
			of the $12$ qubit hyperbolic pentagon code as derived from the 
			graph state in Fig.~\ref{fig:MinimalGraph}.}
	\end{figure}
	
	\subsection{Preparing the logical states}\label{sect:preparingLog0}
	We describe
	how the logical states
	can be prepared experimentally. 
	Throughout we will use 
	the labeling from Fig.~\ref{fig:Minimal}a.
	One starts by preparing the logical zero state $\ket{G^{\text{opt}}_{0000}}$ that is illustrated in Fig.~\ref{fig:LogicalState} and whose formula is given by Eq.~\eqref{eq:logical0X}. 
	Prepare  $\ket{+}^{\otimes 12}$ and apply controlled-Z gates between qubits $u,v \in \E_\PP$,
	\begin{align}
		\ket{G^\text{opt}_{0000}} = \prod_{(u,v) \in \E_\PP} \CZ_{uv}  \ket{+}^{\otimes 12} \label{eq:Gopt}\,.
	\end{align}
	Here ${\E_\PP}$ consists of the $28$ boundary-to-boundary edges from
	Fig.~\ref{fig:LogicalState}, 
	\begin{align}
		\label{eq:EPgates}
		\E_\PP=
		\big\{&(2, 1), (2, 4), (2, 5), (2, 6), (3, 4), (3, 5), (3, 6), \nonumber\\
		&(5, 4), (5, 7), (5, 8), (5, 9), (6, 7), (6, 8), (6, 9), \nonumber\\
		&(8, 7), (8, 10), (8, 11), (8, 12),\nonumber\\
		&(9, 10), (9, 11), (9, 12), \nonumber\\
		&(11, 10), (11, 1), (11, 2), (11, 3), \nonumber\\
		&(12, 1), (12, 2), (12, 3) \big\} \,.
	\end{align}
	With the logical zero state prepared, one can now obtain the remaining logical states by applying logical $\Xbar$~gates stated in Eq.~\eqref{eq:logical0X}. 
	
	\subsection{Encoding circuit}\label{sect:Encircuit}
	In the following paragraphs we describe how to encode an 
	arbitrary four-qubit bulk state 
	into twelve boundary qubits.
	Eq.~\eqref{eq:gateUG} and \eqref{eq:gates}
	describe the associated encoding procedure for $\ket{G^{\text{opt}}}$.
	This yields the encoding unitary ${\op{H}}_G$, which can
	be decomposed into three parts.
	
	\begin{itemize}
		\item[(1)] Unitary $U_1$ prepares the logical zero state and introduces real phases to the logical states as shown in Eq.~\eqref{eq:logicalbasis}. 
		\item[(2)]  Unitary $U_2$  entangles the logical states of the boundary with the computational basis of the bulk.
		\item[(3)]  Unitary $U_3$ decouples the bulk from the boundary, yielding the encoded state on the boundary.
	\end{itemize}
	
	Eq.~\eqref{eq:gates} shows the general form of these unitaries, in particular they can be decomposed in terms of $\CX$ and $\CZ$ gates.
	For the graph in Fig.~\ref{fig:MinimalGraph} they simplify as follows:
	\begin{align}\label{eq:U1}
		U_1 = \!\!\!\!\!\!\!\! \prod_{
			\substack{
				(u,v)  \in \E_\MM, \E_\PP 
			} 
		} \!\!\!\!\!\! \CZ_{uv} = \!\! \prod_{(u,v) \in \E_{\PP}} \!\!\! \CZ_{uv} \,,
	\end{align}
	because our graph does not have interaction between bulk qubits $\E_\MM=0$. The set $\E_{\PP}$ is given in Eq.~\eqref{eq:EPgates}.
	
	Then, $U_2$ can be written as
	\begin{align}\label{eq:U2}
		U_2 =\!\!\!\!\!\!\!\! \prod_{j\in\{\text{A,B,C,D}\}} \!\!\!\! \CXbar_{j}=\!\!\!\!\!\!\!\!   \prod_{(u,v) \in \E_{\PP|\MM}} \!\!\!\!\!\! \CZ_{uv} \,,
	\end{align}
	by decomposing the logical gates 
	$\CXbar$ 
	into 
	two-qubit $\CZ$ gates according to Eq.~\eqref{eq:LogicalX}. 
	Here the set $\E_{\PP|\MM}$ is given by
	\begin{equation}\label{eq:EBgates}\centering
		\begin{aligned}
			\E_{\PP|\MM}=
			\big\{&(\text{A}, 1), (\text{A}, 3), (\text{A}, 4), (\text{A}, 5), (\text{A}, 6) \\
			&(\text{B}, 4), (\text{B}, 6), (\text{B}, 7), (\text{B}, 8), (\text{B}, 9) \\
			&(\text{C}, 7), (\text{C}, 9), (\text{C}, 10), (\text{C}, 11), (\text{C}, 12)\\
			&(\text{D}, 10), (\text{D}, 12), (\text{D}, 1), (\text{D}, 2), (\text{D}, 3) \big\} \,,
		\end{aligned}
	\end{equation}
	
	Finally, $U_3$ reads
	\begin{align}\label{eq:U3}
		U_3 &= \!\!\!\!\!\!\!\! \prod_{j\in\{\text{A,B,C,D}\}} \!\!\!\! {\CZbar}_{j} 
		\Big( 
		\Had^{\otimes 4}\otimes \one_{12} 
		\Big) \\
		&= \!\!\!\!\!\!\!\! \prod_{j\in\{\text{A,B,C,D}\}} \!\! \Big( \prod_{w\in W_j} \! \text{CX}_{jw} \! \prod_{v\in V_j} \! \CZ_{jv}\Big) \Big(\Had^{\otimes 4}\otimes \one_{12}\Big)\,, \nonumber
	\end{align}
	where the logical $\CZbar$ gates are decomposed into $\CX$ and $\CZ$ gates according to Eq.~\eqref{eq:LogicalZ}. Here the sets $V_{j}$ and $W_{j}$ are given by
	\begin{equation}
		\begin{aligned}
			W_{\text{A}}= &\{5\}\,, \quad  &V_{\text{A}}= &\{2,3,4,7,8,9\}\,, \\
			W_{\text{B}}= &\{8\}\,, \quad &V_{\text{B}}= &\{5,6,7,10,11,12\}\,, \\
			W_{\text{C}}= &\{11\}\,, \quad &V_{\text{C}}= &\{1,2,3,8,9,10\}\,, \\
			W_{\text{D}}= &\{2\}\,,  \quad &V_{\text{D}}= &\{1,4,5,6,11,12\}\,.
		\end{aligned}
	\end{equation}
	
	Given the three unitaries written as explicit quantum gates, we can proceed to encode a bulk state. 
	Expand a general bulk state as
	\begin{equation}
		\ket{\phi_{\inc}}= \!\! \sum^1_{a,b,c,d=0} \!\!\!\! c_{abcd}\ket{abcd}\,,
	\end{equation}
	and prepare the 12 boundary qubits in $\ket{+}^{\ot 12}$.
	The application of $U_1$ [c.f. Eq.~\eqref{eq:U1}]  prepares the logical zero state $\ket{G^{\text{opt}}_{0000}}$ from Eq.~\eqref{eq:Gopt} on the boundary,
	\begin{align}
		U_1\big(\ket{\phi_{\inc}}_{\MM} \otimes \ket{+}_{\PP}^{\otimes 12}\big)
		=&\ket{\phi_{\inc}}_{\MM}\otimes\ket{G^{\text{opt}}_{0000}}_{\PP} \,.
	\end{align}
	The subsequent application of $U_2$ [c.f. Eq.~\eqref{eq:U2}] entangles the bulk with the boundary,
	\begin{align}
		\ket{\psi}=&
		U_2 \big(\ket{\phi_{\inc}}_{\MM}\otimes\ket{G^{\text{opt}}_{0000}}_{\PP}\big)
		\nonumber\\ 
		=&\!\!\!\!\!\sum^1_{a,b,c,d=0} \!\!\!\!\! c_{abcd}\ket{abcd}_{\MM}\otimes\ket{G^\text{opt}_{abcd}}_{\PP}\,.
	\end{align}
	Finally, the gate $U_3$ [see Eq.~\eqref{eq:U3}] decouples the bulk from the boundary,
	\begin{align}
		U_3 \ket{\psi} =\ket{+}^{\ot 4}_{\MM}\otimes \!\!\!\!\! \sum^1_{a,b,c,d=0} \!\!\!\!\! c_{abcd} \ket{G^\text{opt}_{abcd}}_{\PP}. 
	\end{align}
	The encoded state is then
	\begin{align}
		\ket{\phi_{\enc}}= \sum^1_{a,b,c,d=0} c_{abcd} \ket{G^{\text{opt}}_{abcd}}\,,
	\end{align}
	with the computational basis $\ket{abcd}$ mapped to the logical basis $\ket{G^{\text{opt}}_{abcd}}$ .
	
	\begin{figure*}[ht]
		\centering
		\includegraphics[width=0.9\textwidth]{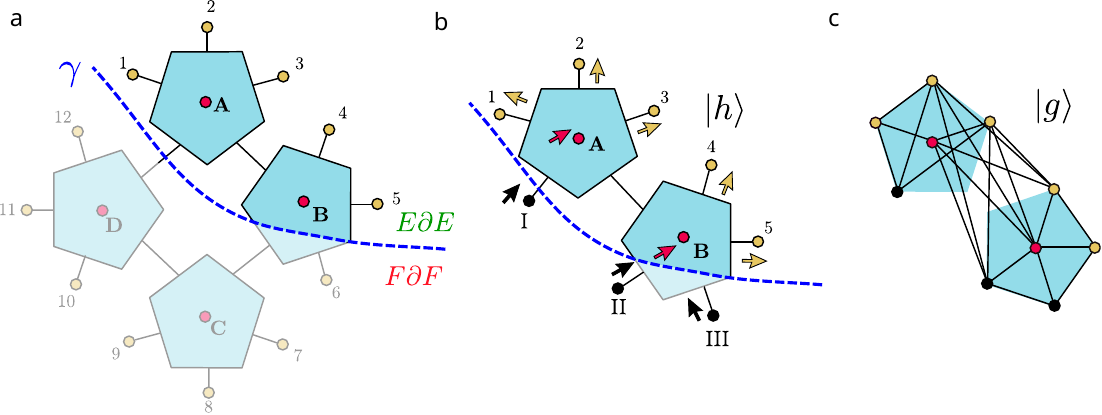}
		\caption{
			\label{fig:rec}
			\textbf{Decoding scheme for $12$ qubit hyperbolic pentagon code.} A cut for which an isometry from the red qubits of $\regE$ to the golden qubits of $\regdE$ exist (Fig.~\ref{fig:rec}a). This isometry is illustrated in Fig.~\ref{fig:rec}b. 
			After that and applying a set of Hadamards, we obtain the graph code Fig.~\ref{fig:rec}c. This code can be used to recover partially the information of the original encoding in Fig.~\ref{fig:rec}a.}
	\end{figure*}

	\subsection{Partial decoding circuit}\label{sect:Pdecodingcircuit}
	The graph code (Fig.~\ref{fig:MinimalGraph}) gives information on how to perform the encoding.
	However, it provides little intuition on realizing a partial decoding operation, i.e., recovering a part of the bulk from its nearby boundary. 
	The geometry of the tensor network indicates what partial recovery processes are possible in general as discussed in Ref.~\cite[Section 5.3]{Pastawski_2015}.
	
	For our toy model, Fig.~\ref{fig:rec} illustrates how to perform a partial recovery operation for a specific choice of bulk and boundary regions. 
	Consider the specific cut $\gamma$ from Fig.~\ref{fig:rec}a which separates two regions: $\regdE \cup \regE$ and $\regdF \cup\regF$, where $\regE \cup \regF$ are bulk qubits (red) and $\regdE \cup\regdF$ boundary qubits (gold). 
	Further, we define the black qubits labeled by I, II and III as those associated to each tensor network leg crossed by the cut $\gamma$.
	The cut is placed in such a way that the two contracted AME states in Fig.~\ref{fig:rec}b act as an isometry $T_{\op{h}}$ from $\regE \cup \gamma$ to $\regdE$. 
	As shown in Eq.~\eqref{eq:partialiso}, 
	one can use $T^{\dagger}_{\op{h}}$ 
	to recover the bulk qubits from $\regE$ 
	by only reading the boundary $\regdE$.
	
	The isometry $T_{\op{h}}$ can be represented as a quantum state $\ket{h}$ obtained by a contraction of two AME states in Fig.~\ref{fig:rec}b.
	This state $\ket{h}$ can also be mapped to a graph state by applying local Clifford operations.
	Again, this graph can be optimized with respect to the number of
	edges and locality leading to the state $\ket{g}$  illustrated in Fig.~\ref{fig:rec}c.
	Appendix~\ref{app:Pdec} shows how to correct the  $U^{\dagger}_g$ for this local Clifford optimized code.
	
	The decoding procedure contains then the following steps:
	\begin{itemize}
		\item[1.] Apply the corresponding local Clifford before $U^{\dagger}_g$.
		
		\item[2.] Apply the same procedure as described in \hyperref[sect:Encircuit]{\textit{Encoding circuit}}, but for the graph $\ket{g}$ and in a reverse order since $U^{\dagger}_g=U_3^\dagger U_2^\dagger U_1^\dagger$. 
		\item[3.] Apply the corresponding local Clifford after $U^{\dagger}_g$.
	\end{itemize}
	
	For the graphs $\ket{G^{\text{opt}}}$ and $\ket{g}$ illustrated in Fig.~\ref{fig:MinimalGraph} and Fig.~\ref{fig:rec}c respectively, the unitary gate $U^{\dagger}_{\op{g}}$ is modified and given by ${\widetilde{U}^\dagger_{\op{h}}}={Z_BU^\dagger_{\op{g}}}$ (see the end of Appendix~\ref{app:Pdec}). First, $U^{\dagger}_3$ can be written as
	\begin{align} \label{eq:U3exp}
		U^{\dagger}_3&=\Big(H^{\otimes 5}_{\regE\gamma}\otimes \one^{\otimes 12}_{\PP} 
		\Big) \!\!\! \prod_{j\in \{\text{A,B},\text{I,II,III}\}} \!\!\!\!\! {\CZbar}_{j} \nonumber\\
		&=\Big({H}^{\otimes 5}_{\regE\gamma}\otimes \one^{\otimes 12}_{\PP} 
		\Big)  \\ & \cdot \prod_{j\in\{\text{A,B},\text{I,II,III}\}} \!\!  \Big( \! \prod_{\omega\in \Omega_j} \CZ_{j\omega} \! \prod_{w\in W_j}  \text{CX}_{jw} \! \prod_{v\in V_j} \CZ_{jv} \! \Big) \,, \nonumber
	\end{align}
	where the gates $\CZbar$ decompose into $\text{CX}$ and $\CZ$ according to
	\begin{equation}\label{eq:LogicalXZ}
		\!\!\begin{matrix}[cc ccccc c cc ccccc]
			&&\text{\footnotesize1}&\text{\footnotesize2}&\text{\footnotesize3}&
			\text{\footnotesize4}&\!\!\text{\footnotesize5}
			&\; && &\text{\footnotesize1}&\text{\footnotesize2}&\text{\footnotesize3}&
			\text{\footnotesize4}&\!\!\text{\footnotesize5} \\
			\thickbar{X}_{{\op{I}}}&=& \pZ & \pZ & \pZ & \one & \one \,,
			& \; \thickbar{Z}_{\op{I}}&=& &\pY & \pY & \pX & \one & \one \,,\\
			\thickbar{X}_{\op{A}}&=& \pZ & \one & \pZ & \pZ & \one \,,
			& \; \thickbar{Z}_{\op{A}}&=& &\pZ & \pX & \pX & \one & \one \,,\\
			\thickbar{X}_{\op{III}}&=& \one & \one & \one & \one & \pZ \,,
			& \; \thickbar{Z}_{\op{III}}&=&-\!\!\!\!&\pZ & \pY & \pY & \pY & \pY \,,\\
			\thickbar{X}_{\op{II}}&=& \one & \pZ & \pZ & \one & \one\,,
			& \; \thickbar{Z}_{\op{II}}&=&-\!\!\!\!&\pY & \pX & \pZ & \pY & \pZ \,,\\
			\thickbar{X}_{\op{B}}&=& \one & \pZ & \pZ & \pZ & \pZ\,,
			& \; \thickbar{Z}_{\op{B}}&=& -\!\!\!\!&\pZ & \pY & \pY & \pX & \pZ\,.\\
		\end{matrix}
	\end{equation}
	In Eq.~\eqref{eq:U3exp}, the sets $\Omega_{j}, V_{j}, W_j$ are given by
	\begin{equation}
		\!\!\begin{aligned}
			\Omega_{\text{I}}=& \{2\}\,,  
			&W_{\text{I}}= &\{1,2,3\}\,, 
			&V^{\text{I}}_{z}=& \{1\}\,,
			\\
			\Omega_{\text{A}}=& \varnothing\,,   
			& W_{\text{A}}= &\{2,3\}\,, \;\; 
			&V_{\text{A}}=& \{1\}\,, \\
			\Omega_{\text{III}}=& \{2\}\,,    & W_{\text{III}} =&\{2,3,4,5\}\,,   &V_{\text{III}}=& \{1,3,4,5\}\,, \\
			\Omega_{\text{II}}=& \varnothing \,,   & W_{\text{II}}= &\{1,2,4\}\,, \  &V_{\text{II}}=& \{1,3,4,5\}\,, \\
			\Omega_{\text{B}}=& \varnothing \,,  & W_{\text{B}}= &\{2,3,4\}\,, \,\,  &V_{\text{A}}=& \{1,2,3,5\}\,. 
		\end{aligned}
	\end{equation}
	
	Note that since the logical $\Zbar$~gates shown in Eq.~\eqref{eq:LogicalXZ} have extra phases, we need to introduce a new set of gates described by $V^{j}_{\omega}$. 
	After expressing the logical $\Zbar$~gates via X and Z and because of $\pY=i\pX\pZ$,
	only $Z_\text{I}$ and $Z_{\text{III}}$ will have extra phases.
	
	Second, $U^{\dagger}_2$ can be written as
	\begin{equation}
		U^{\dagger}_2=\!\!\!\!\prod_{j\in \{\text{A,B,I,II,III}\}} \!\!\!\! {\CZbar}_{j}= \!\!\!\!\!\!\prod_{(u,v)\in \E_{\regdE|\regE\gamma}} \!\!\!\!\CZ_{uv} \,,
	\end{equation}
	by decomposing $\CXbar$ into $\CZ$ gates according to Eq.~\eqref{eq:LogicalXZ}. 
	Here the set $\E_{\regdE|\regE\gamma}$ is given by
	\begin{equation}\label{eq:EBgates2}\centering
		\begin{aligned}
			\E_{\regdE|\regE\gamma}=\big\{ &(\text{I},1), (\text{I},2), (\text{I},3), (\text{A},1), (\text{A},3), (\text{A},4), \\ & (\text{III},5),
			(\text{II},2), (\text{II},3), (\text{B},2),\\ &  (\text{B},3), (\text{B},4), (\text{B},5)\big\} \,.
		\end{aligned}
	\end{equation}
	The last unitary $U^{\dagger}_1$ is
	\begin{align}
		\label{eq:decoding1}
		U^\dagger_1= \!\!\!\! \prod_{(u,v) \in \E_{\regdE},\E_{\regE\gamma}} \!\!\!\! \CZ_{uv} \,,
	\end{align}
	with the sets $\E_{\regE\gamma}$ and $\E_{\regdE}$ given by
	\begin{align}
		\E_{\regE\gamma}=&
		\big\{(\text{II}, \text{A}),(\text{II}, \text{III}),(\text{B}, \text{A}),(\text{B}, \text{III}),(\text{B}, \text{II})\big\}\,, \\
		\E_{\regdE}=&\big\{(2, 1),(4, 2),(4, 3),(5, 4)\big\} \,.
	\end{align}
	
	Starting from the encoded state $\ket{\phi_{\enc}}$ we introduce five extra qubits, two red from the bulk region $\regE$ and three black from the cut (Fig.~\ref{fig:rec}b), such that
	\begin{align}\label{eq:decoding0}
		\ket{\psi_1}= \ket{+}^{\otimes 5}_{\regE\gamma} \otimes \ket{\phi_{\enc}}_{\PP}
		\,. 
	\end{align}
	For our graphs $\ket{G^{\text{opt}}}$ and $\ket{g}$, we need to apply local Clifford operations after $U^\dagger_{\op{g}}$ and not before, since $\widetilde{U}^\dagger_{\op{h}}={Z_BU^\dagger_{\op{g}}}$. %
	Therefore, apply $U^\dagger_{\op{g}}=U^\dagger_1 U^\dagger_2 U^\dagger_3$ to obtain
	\begin{equation}\label{eq:decoding-1}
		\ket{\psi_0}=\Big(U^{\dagger}_1 U^{\dagger}_2 U^{\dagger}_3 \otimes \one_{\regdF} \Big)\ket{\psi_{1}}\,,
	\end{equation}
	and apply $Z_B$ to obtain the decoded state,
	\begin{equation}    
		\ket{\psi_{\text{dec}}}=Z_B\ket{\psi_0}\,.
	\end{equation}
	
	The state $\ket{\psi_{\text{dec}}}$ has the same reduction on $\regE$ as $\ket{\phi_{\inc}}$ [see Eq.~\eqref{eq:partialdecoding}]. Thus the bulk information on $\regE$ can be recovered from the state $\ket{\psi_{\text{dec}}}$ by only reading its nearby boundary $\regdE$. Thus the sequence of unitary gates Eq.~\eqref{eq:decoding0} to Eq.~\eqref{eq:decoding-1} realizes partial decoding from the encoded state and hence can be used to demonstrate holographic bulk reconstruction.
	
	In summary, we need $12$ qubits to prepare a code state, $16$ qubits to encode an arbitrary state and $17$ qubits to perform partial recovery. 
	
	A circuit to perform encoding followed by partial decoding can also be realized with $17$ qubits, as shown in Fig.~\ref{fig:Expsetup}.
	The encoding needs $12$ boundary qubits and $4$ bulk qubits initialized in a product state and $\ket{\phi_{\text{in}}}$, respectively. The result of the encoding is a product state in the bulk and an encoded state in the boundary $\ket{\phi_{\text{enc}}}$.
	The partial decoding needs $12$ qubits from the encoded state in the boundary, $2$ qubits (A and B) from the bulk region $\regE$ that we aim to recover and $3$ qubits (I, II and III) from the cut $\gamma$. 
	The remaining bulk qubits (C and D) can be recycled to act as qubits on the cut for the partial decoding. Fig.~\ref{fig:Expsetup} shows how the qubits C and D are reused as I and II. Note that we still need an additional qubit III to perform the partial decoding;
	we introduce such a qubit in the centre.
	
	We emphasise that the experimental setup in Fig.~\ref{fig:Expsetup}  describes the case of performing a partial decoding when we have previously performed a general encoding.
	For performing a partial decoding given a $12$-qubit boundary state, we have to add $3$ extra qubits as shown in Fig.~\eqref{fig:rec}b.

	\begin{figure}
		\centering
		\includegraphics[width=0.33\textwidth]
		{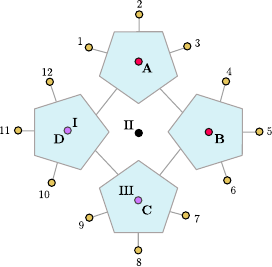}
		\caption{\textbf{Experimental setup with $17$ qubits.} 
			The bulk qubits A, B, C and D are encoded into the boundary qubits $1, \, \dots,\, 12$. After encoding,  we reuse the qubits C (now III) and D (now I) and add an extra qubit III to perform the partial decoding.
			This procedure recovers the bulk degrees of freedom A and B from the boundary given by the qubits $1$ to $5$. \label{fig:Expsetup}}
	\end{figure}

	\subsection{Experimental feasibility}
	\label{sect:expfeas}
	For the implementation of the graph states presented in this work the ability to entangle arbitrary qubits is essential. In many state-of-the-art approaches, the interaction between qubits is local, which constrains the connectivity of the artificial quantum system. Hence, the long-range entangling gates have to be decomposed into local entangling gates which then increases the number of gates necessary to realize an equivalent circuit. However, several platforms are outstanding with their ability to generate non-local connectivity between qubits and hence allow to prepare the entangled graph states. Platforms which provide such connectivity are trapped ions~\cite{Ringbauer2022}, Rydberg arrays~\cite{Bluvstein2022}, (artificial) atoms coupled to a cavity~\cite{Devoret2013, Periwal2021}.   
	In the case of trapped ions, the long-range interaction between two qubits can be achieved by either using the phonon degrees of freedom in an ion crystal~\cite{Porras2004} or a shuttling approach~\cite{PhysRevLett.109.080501}, i.e., moving the ions next to each other and then entangle them. Similarly, recent experimental progress in Rydberg atom arrays allow for entangling arbitrary pairs of atoms by shuttling the atoms, bringing two atoms next to each other and then entangle them by using the Rydberg blockade mechanism. 
		
	In the case of (artificial) atoms coupled to a cavity, the two typical platforms are superconducting qubits coupled to a microwave cavity or Rubidium atoms coupled to a cavity. In superconductor based platforms Josephson junctions are used to form an artificial two level system.~\cite{Devoret2013}, whereas in the case of atomic systems one frequently uses internal states of the atom~\cite{Periwal2021}. The range of the interaction can be engineered by exploiting that a photon can travel along a cavity, while the atoms are connected to the cavity. Controlling the coupling between the (artificial) atom and the cavity then allows to engineer long-range interactions.
	
	In principle, all of these three systems trapped ions, Rydberg arrays or (artificial) atoms coupled to a cavity are able to perform arbitrary local unitary operations and one entangling operation between two arbitrary qubits, which is sufficient to perform arbitrary unitary operations between two qubits~\cite{Lloyd1995}. For example, Rydberg atom arrays recently demonstrated the generation of 12-qubit cluster state and a seven-qubit Steane code with related stabilizer measurements~\cite{Bluvstein2022}. Similar results were achieved in ion trap setups, where for example a 7 qubit color code~\cite{Nigg2014} or a 9-qubit Bacon Shore code state~\cite{Egan2021} were realized. 
	
	Ref.~\cite{Quantinuum} bench-marked a trapped-ion setup with single-qubit, two-qubit gate and measurement fidelity of $0.99994(3)$, $0.9981(3)$ and $0.9972(5)$, respectively.
	Using such a setup, we make a simple estimation of the state fidelity.
	The logical zero state is obtained with a fidelity of $0.947(8)$ from the gate sequence described in Eq.~\eqref{eq:Gopt}. 
	Other approaches to prepare the logical zero state require non-destructive stabilizer measurements~\cite{Abobeih_2022} and result in a fidelity of $0.88(2)$. 
	The state fidelity shows a significant advantage in preparing the logical zero state as a graph compared to stabilizer measurements. 
	This estimation is developed with more details in Appendix~\ref{app:fidelity}.
	
	\subsection{Related work}
	
	Ref.~\cite[Section 5.7-8]{Pastawski_2015} already highlighted that the hyperbolic pentagon can be formulated in the stabilizer formalism.
	The work of Ref.~\cite{PhysRevA.101.042305} then explored more explicitly the stabilizer formulation, taking advantage of index contractions formulated in terms of Bell state projections and obtaining corrections to the Ryu Takanayagi formula for entangled input states. 
	The resulting codes are not yet optimized over local Clifford operations to reduce experimental requirements.
	Ref.~\cite{Beigi_2011} introduced the concatenation of quantum codes through their graph state representation. 
	This method of concatenation is not directly applicable to a more general tensor network such as the hyperbolic pentagon code, as in this case more general index contractions are required.

	\section{Discussion}
	In this work we provided a systematic method to represent the hyperbolic pentagon (HaPPY) code as a stabilizer graph code.
	This allows to engineer as of now theoretical models of holography in artificial quantum systems. 
	Interestingly, the formulation as a graph code applies to any code defined through a tensor network with stabilizer states as building blocks.
	Furthermore, the method is not restricted to qubits, 
	but, with suitable generalizations, also applies to qudits in prime dimensions.
	This provides us with the tools to engineer other codes in the same manner, e.g., the holographic state~\cite{https://doi.org/10.48550/arxiv.1502.06618}, holographic CSS codes~\cite{Harris_2018}, 
	and random stabilizer tensor networks~\cite{PhysRevLett.125.241602}.
	
	Regarding the scalability of our proposal,
	the local Clifford optimization as performed here is a limiting factor to find experimentally suitable formulations for larger instances of the hyperbolic pentagon code. 
	However, an inspection of the building blocks graphs before the contraction can provide intuition on how the final contracted graph code might look like.
	Reducing the number of entangling gates and keeping them short range is crucial for current noisy intermediate-scale quantum (NISQ) devices.
	
	Several open questions of interest remain.
	It is unclear whether the usage of symmetric building blocks always helps in reducing the local Clifford gate complexity for optimizing the number and range of edges. 
	It is also an unknown whether there is an efficient iterative procedure to build a larger holographic graph layer by layer.
	Finally, it would be interesting to understand 
	the advantages and disadvantages of experimentally implementing the hyperbolic pentagon code in prime dimensions.
	
	\subsection{Data availability}
	
	There is no additional data.
	
	\subsection{Code availability}
	
	The code that supports the findings is available in \url{https://github.com/ganglesmunne/Engineering_holography}.
	
	\subsection{Acknowledgments} 
	We want to thank Máté Farkas, Michał Horodecki, Maciej Lewenstein, Pavel Popov, Anna Garcia Sala, Adam Burchardt, Martin Ringbauer, Markus Grassl and Karol Życzkowski for comments and fruitful discussions. 
	
	G.A.M. and F.H. are supported by the Foundation for Polish Science through TEAM-NET (POIR.04.04.00-00-17C1/18-00). 
	Much of this work was done while F.H. was working at the Jagiellonian University in Kraków.
	
	ICFO group acknowledges support from: ERC AdG NOQIA; MCIN/AEI (PGC2018-0910.13039/501100011033, CEX2019-000910-S/10.13039/501100011033, Plan National FIDEUA PID2019-106901GB-I00, Plan National STAMEENA PID2022-139099NB-I00 project funded by MCIN/AEI/10.13039/501100011033 and by the “European Union NextGenerationEU/PRTR” (PRTR-C17.I1), FPI); QUANTERA MAQS PCI2019-111828-2; QUANTERA DYNAMITE PCI2022-132919 (QuantERA II Programme co-funded by European Union’s Horizon 2020 program under Grant Agreement No 101017733), Ministry of Economic Affairs and Digital Transformation of the Spanish Government through the QUANTUM ENIA project call - Quantum Spain project, and by the European Union through the Recovery, Transformation, and Resilience Plan - NextGenerationEU within the framework of the Digital Spain 2026 Agenda; Fundació Cellex; Fundació Mir-Puig; Generalitat de Catalunya (European Social Fund FEDER and CERCA program, AGAUR Grant No. 2021 SGR 01452, QuantumCAT U16-011424, co-funded by ERDF Operational Program of Catalonia 2014-2020); Barcelona Supercomputing Center MareNostrum (FI-2023-1-0013); EU Quantum Flagship (PASQuanS2.1, 101113690); EU Horizon 2020 FET-OPEN OPTOlogic (Grant No 899794); EU Horizon Europe Program (Grant Agreement 101080086 — NeQST), ICFO Internal “QuantumGaudi” project; European Union’s Horizon 2020 program under the Marie Sklodowska-Curie grant agreement No 847648; “La Caixa” Junior Leaders fellowships, La Caixa” Foundation (ID 100010434): CF/BQ/PR23/11980043. Views and opinions expressed are, however, those of the author(s) only and do not necessarily reflect those of the European Union, European Commission, European Climate, Infrastructure and Environment Executive Agency (CINEA), or any other granting authority. Neither the European Union nor any granting authority can be held responsible for them.

	\subsection{Author contributions}
	
	All authors contributed to the scientific process and the writing of the article. FH and VK developed the conception of the project and the main theoretical idea. GAM and FH mainly worked out the theoretical details, while GAM and VK specifically considered the experimental feasibility.
	
	\subsection{Competing interests}
	The Authors declare no Competing Financial or Non-Financial Interests.
	\appendix\onecolumngrid

	\section{Index contraction on stabilizer states}\label{app:stabilizer}

	\begin{restatable}{fact}{factStabBell}\label{factStabBell}
		Let $\ket{\psi}$
		be an $m$-qubit stabilizer state and contract two indices of it. Then the state after the contraction $\ket{\chi}$ will be a stabilizer state.
	\end{restatable}

	\begin{proof}
		
		Eq.~\eqref{eq:Bellcontr} shows that contracting two indices of $\ket{\psi}$ corresponds to a projection onto the Bell state 
		$\ket{\phi^+}$ followed by a partial trace. 
		Since $\ket{\psi}$ and $\ket{\phi^+}$ are stabilizer states, the projection leads to another stabilizer state (see Ref.~\cite[Chapter 10.5.3]{nielsen_chuang_2010}).
		Finally, also performing a partial trace of a stabilizer state also leads to a stabilizer state (see Ref.~\cite[Section 2]{Audenaert_2005}) and so, $\dyad{\chi}$ is stabilizer.  This ends the proof.
	\end{proof}

	\noindent{\bf The generators of $\ket{\chi}$.} Ref.~\cite[Chapter 10.5.3]{nielsen_chuang_2010} shows a method to obtain the stabilizer group of a state after the projection onto the Bell state at positions $i$ and $j$.
	Recalling that the generators of the Bell state  are $\{\pX\ot \pX, \pZ \ot \pZ\}$, this method requires:
	\begin{itemize}
		\item[(i)] Perform elementary row operations on the check matrix of the $m$-qubit state $\ket{\psi}$ which leads to, at most, one row anti-commuting with $\pX_i \ot \pX_j$ and another row anti-commuting with $\pZ_i \ot \pZ_j$.
		\item[(ii)] Substitute $\pX_i \ot \pX_j$ and $\pZ_i \ot \pZ_j$ for their corresponding anti-commuting rows to obtain the check matrix of the projected state $\ket{\phi^+} \ot\ket{\chi}$ [see Eq.~\eqref{eq:Bellcontr}].  
	\end{itemize}
	By removing the two columns that index the projected qubits, followed by removing the two rows which are now linearly dependent, one obtains the generators of the contracted state $\ket{\chi}$.  This method performs elementary row operations and partial trace over the Bell state.  Even with a trivial vector phase $\omega=\bmR{0}$, both may introduce non-trivial phases in the parity check matrix.
	
	We now describe how to calculate these changes in $\omega$:
	
	\smallskip
	\noindent {\bf Changes in $\omega$ under elementary row operations}.
	When two rows are added under the elementary row operations of (i), a non-trivial phase may appear.
	This corresponds to the multiplication of the two generators, 
	whose one-qubit Pauli matrices do not commute in $2 +4j$ with $j\in \N$ positions.
	This phase must necessarily be real since $-\one \notin S$. 
	
	\smallskip
	\noindent {\bf Changes in $\omega$ after partial tracing the Bell state}.
	Let us write the generators of $\ket{\chi}$ 
	as $g_i = \alpha_i \tilde g_i$  where $\alpha_i \in \{-1,1\}$ carries the phase.
	The stabilizer of the Bell state is 
	$\{\one \otimes \one,\, 
	\pX\ot\pX,\,
	\pZ\ot\pZ,\,
	-\pY\ot\pY \}$. 
	Then we can list the generators of $\ket{\phi^+}\otimes \ket{\chi}$ as on the left-hand side of Eq.~\eqref{eq:phasesgenerator}.
	\begin{equation}\label{eq:phasesgenerator}
		\begin{matrix}[ccc|c]
			\ket{\phi^{+}}      &\ot&\!\!\!\! \ket{\chi} & \, \phantom{-}\\ \hline
			\pX \ot \pX &\ot    &\!\!\!\! \tilde g_1 & \phantom{-}\alpha_1\\
			\pZ \ot \pZ & \ot   &\!\!\!\! \tilde g_2 & \phantom{-}\alpha_2\\
			\pY \ot \pY & \ot   &\!\!\!\! \tilde g_3 & \phantom{-}\alpha_3\\  
			\vdots & \vdots     &\!\!\!\!  \vdots  & \vdots \\
			\one \ot \one & \ot &\!\!\!\! \tilde g_m & \phantom{-}\alpha_m\\
			\multicolumn{3}{c}{+ \,\text{two extra rows}}
		\end{matrix} 
		\quad
		\xrightarrow{\text{Partial trace over}\,\,\ket{\phi^+}}
		\quad
		\begin{matrix}[c|c]
			\ket{\chi} & \phantom{-} \\ \hline
			\tilde g_1 & \phantom{-}\alpha_1\\
			\tilde g_2 & \phantom{-}\alpha_2\\
			\tilde g_3 & -\alpha_3\\  
			\vdots & \phantom{-} \vdots \\
			\tilde g_m & \phantom{-}\alpha_m \\
			\multicolumn{2}{c}{}
			
		\end{matrix}
	\end{equation} 
	
	Here the two extra rows are those that are linearly dependent with respect to the remaining generators after removing the two columns that index $\ket{\phi^{+}}$.
	To partial trace over the subsystem with the Bell state, one needs to remove its stabilizer group from the generators of $\ket{\phi^+}\otimes \ket{\chi}$, in addition to removing the two linear dependent rows leading to left-hand side of Eq.~\eqref{eq:phasesgenerator}.
	Since one of the stabilizers of $\ket{\phi^{+}}$ has a sign $-\pY\ot\pY$, whenever one removes the element $\pY\ot\pY$ from the Bell state subsystem, one needs to change the sign of corresponding generator from $\ket{\chi}$. We can see this in the case of the generator $g_3$ on the left-hand side of Eq.~\eqref{eq:phasesgenerator}.

	\section{Index contraction on graph states}\label{app:graphH}

	\begin{restatable}{fact}{factHal}
		\label{factHal}
		Let $\ket{G}$ be graph state on $m$ qubits. Contracting two indices $i$ and $j$ from $\ket{G}$ yields another graph state $\ket{G'}$ on $m-2$ qubits, up to a single layer of $\Had$~gates followed by a layer of $\pZ$~gates. 
	\end{restatable}
	
	\begin{proof}
		Appendix~\ref{app:stabilizer} shows how to obtain the generators of the state after an index contraction. 
		This method requires
		elementary row operations on the enlarged 
		check matrix of $\ket{G}$ [left-hand side of Eq.~\eqref{eq:gauss}], 
		followed by the elimination of the corresponding two rows and columns. 
		
		To transform this new check matrix into a graph state, row-reduce its $\pX$-part via elementary row operations, 
		\begin{equation}
			\begin{pmatrix}[c|c] 
				\one_{m} & \Gamma\\
			\end{pmatrix}\,
			\xrightarrow{\substack{
					\text{Index contraction} \\ 
					\,\,
					\text{and row operations}}
				\,\,
			}
			\,\begin{pmatrix}[cc|cc]\label{eq:gauss}
				\one_{p} & A & B & C   \\ 
				0 & 0_{q} & D & E
			\end{pmatrix}\quad\text{with}\quad p+q=m-2\,.
		\end{equation}
		For a better representation we ordered the columns of the check matrix without changing the qubit labeling in the right-hand side of Eq.~\eqref{eq:gauss}. 
		This new check matrix satisfies two properties:
		\begin{itemize}
			\item[(i)] 
			The matrix $E$ is full rank.
			\item[(ii)] 
			Every row contains exactly an even number of positions where both the $X$ and the $Z$ entries are $1$;
			in other words,
			all generators contain an even number of $\pY$'s. 
		\end{itemize} 
		
		One can see (i) by contradiction: 
		suppose that $E$ is not full rank. Then by multiplying generators of $(0 \, 0_{q} \,|\, D \, E)$ 
		one can obtain an element 
		$t=(0 \, 0 \,|\, T \, 0)$. 
		However, from Eq.~\eqref{eq:pc_comm} one sees that this element $t$ commutes with $(\one_{p} A \,|\, B\, C)$ if and only if $T=0$. 
		Thus for these elements to commute, $t$ must be trivial $t=(00|00)$, which cannot happen since the generators are independent among each other. 
		Therefore $E$ is full rank. 

		To show (ii), recall that the generators of a graph state $\ket{G}$ contain Pauli $\pX$ and $\pZ$ only. 
		Then elementary row operations on the check matrix can produce an even number of $\pY$s only,
		for the stabilizer elements to remain hermitian.
		The same argument applies to the state after the elimination of those columns leading the the check matrix of the contracted state.

		Let us now row-reduce $E$ via row operations. This elimination yields the left-hand side of Eq.~\eqref{eq:gauss2}. Since the generators commute, ${D}=A^{\intercal}$
		by Eq.~\eqref{eq:pc_comm}. 
		Finally, row operations that involve 
		$(0 \, 0\,|\,D\,\one_{q})$
		set $C$ to zero. 
		In summary
		\begin{equation}\label{eq:gauss2}
			\begin{pmatrix}[cc|cc]
				\one_{p} & A & B & C \\ 
				0 & 0_{q} & {D} & \one_{q}
			\end{pmatrix} 
			\,\,
			\xrightarrow{\text{row operations}}
			\,\,
			\begin{pmatrix}[cc|cc]
				\one_{p} & A & \widetilde{B} & 0 \\ 
				0 & 0_{q} & A^\intercal & \one_{q}
			\end{pmatrix}\,.
		\end{equation}
		Here, $\widetilde B$ is symmetric since the generators commute.
		Also, $\widetilde B$ contains zeros on the diagonal, as the generators 
		can only contain an even number of $\pY$s.
		
		Applying $\Had^{\otimes q}$ to the last $q$ qubits maps $\pX$ to $\pZ$ and $\pZ$ to $\pX$.
		This application of Hadamard gates exchanges the last $q$ columns of the $\pX$ and $\pZ$ parts in the check matrix, yielding
		\begin{equation}\label{eq:graphproof}
			\begin{pmatrix}[cc|cc|c]
				\one_{p} & 0 & \widetilde{B} & A & \multirow{2}{*}{$\omega$} \\
				0 & \one_{q} & A^\intercal & 0 &
			\end{pmatrix}\,.
		\end{equation}
		Here we have included $\omega$ since previous row operations may have introduced new signs to the generators. 
		It is clear that Eq.~\eqref{eq:graphproof} is the check matrix of a graph state $(\one|\Gamma)$ with a non-trivial phase vector~$\omega$. 
		Since there are $\pX$'s in the diagonal, 
		the application of $(\pZ^{\omega_1} \otimes \dots \otimes \pZ^{\omega_{p+q}})$ corrects $\omega$ to  $(0\dots0)$.
		The check matrix of Eq.~\eqref{eq:graphproof} with a trivial phase vector is the check matrix of a graph state $\ket{G'}$. 
		This ends the proof.
	\end{proof}
	
	\section{Logical states and isometry}\label{app:logIso}

	\begin{restatable}{fact}{factstate}\label{factstate}
		The graph state $\ket{G}$ 
		of Eq.~\eqref{eq:check0} can be written as
		\begin{equation}\label{eq:graphst_app}
			\ket{G} = \sum^1_{i_1,\dots,i_k =0}
			\ket{i_1 \dots i_k}_\MM \ot
			\ket{G_{i_1 \dots i_k}}_\PP 
			\,.
		\end{equation}
		Basis elements of the code space $\ket{G_{i_1 \dots i_k}}$ with $i_1,\dots,i_k\in \{0,1\}$ can be defined as
		\begin{equation}\label{eq:logicalbasis_app}
			\ket{G_{i_1 \dots i_k}}= (-1)^{\sum_{(u,v)\in \E_\MM} i_vi_u}
			\prod^{k}_{r=1} \Xbar^{i_r}_r \ket{G_{0 \dots 0}} \,,
		\end{equation}
		where
		$\ket{G_{0\dots0}}$  
		is the logical zero state and $\{\Xbar_r\}^k_{r=1}$ are the logical $\Xbar$~gates,
		\begin{equation}\label{eq:logical0X_app}
			\ket{G_{0 \dots 0}}  = \!\!\prod_{(u,v) \in \E_\PP} \!\! \CZ_{uv} \ket{+}^{\otimes n} \,,\,\,\, \Xbar_r =\bigotimes^{n}_{s=1} \pZ^{B_{rs}}\,.
		\end{equation}
	\end{restatable}
	
	\begin{proof}
		In terms of controlled-$\pZ$ gates, the graph state $\ket{G}$ can be written as
		$\ket{G}=\prod_{(u,v)\in E} \CZ_{uv} \ket{+}^{\ot (n+k)}$
		where $E$ is the edge set of the graph. Since $\CZ$ gates on different parties commute, we can order them as follows,
		\begin{equation}
			\ket{G}=  \Big( \prod_{(u,v)\in \E_\MM} \CZ_{uv}   \Big)  \Big(\prod^k_{r=1}\prod^n_{s=1}  \CZ_{rs}^{B_{rs}} \Big)  \Big( \prod_{(u',v')\in \E_\PP}   \CZ_{u'v'} \Big)   \ket{+}^{\otimes k}_{\MM} \ot \ket{+}^{\otimes n}_{\PP} \;.
		\end{equation}
		Define $\ket{G_{0\dots0}} :=\prod_{(u',v')\in \E_\PP}  \CZ_{u'v'}\ket{+}^{\otimes n}$ and expand the bulk-boundary $\CZ$ gates to obtain
		\begin{align}\label{applog1}
			\ket{G}&=\prod_{(u,v)\in \E_\MM} \CZ_{uv} \sum^1_{i_1,\dots,i_k=0}   \Big( \prod^n_{s=1} \CZ^{B_{1s}}_{1s}\ket{i_1}_{\MM_1} 
			\otimes \dots  \otimes \prod^n_{t=1} \CZ^{B_{kt}}_{kt}\ket{i_k}_{\MM_k} \otimes  \ket{G_{0\dots0}}_\PP\Big) \nonumber\\ 
			&=
			\prod_{(u,v)\in \E_\MM} \CZ_{uv}  
			\sum^1_{i_1,\dots,i_k=0} \Big(\ket{i_1 \dots i_k}_\MM  
			\ot 
			\prod^k_{r=1} \bigotimes^n_{s=1} \pZ^{i_r{B_{rs}}} \ket{G_{0\dots0}}_\PP\Big)\,,
		\end{align}
		where we used that $\CZ_{rs} \big(\ket{i_r}\otimes \ket{i_s}\big)=\ket{i_r} \otimes \pZ^{i_r}\ket{i_s}$. 
		Define 
		$\Xbar_r :=\bigotimes^n_{s=1} \pZ^{{B_{rs}}}$ 
		and move the $\CZ$ gates into the sum,
		\begin{equation}\label{applog2}
			\ket{G}=\sum^1_{i_1,\dots,i_k=0} \Big( \prod_{(u,v)\in \E_\MM} \CZ_{uv} \ket{i_1 \dots i_k}_\MM \otimes \prod^k_{r=1} \Xbar^{i_r}_r \ket{G_{0\dots0}}_\PP \Big) \,.
		\end{equation}
		Note that when $\CZ$ gates act on the computational basis, they only add real global phases,
		\begin{equation}\label{eq:phases}
			\prod_{(u,v)\in \E_\MM} \CZ_{uv} \ket{i_1 \dots i_k}=\prod_{(u,v)\in \E_\MM} (-1)^{i_vi_u} \ket{i_1 \dots i_k}=(-1)^{\sum_{(u,v)\in \E_\MM} i_vi_u} \ket{i_1 \dots i_k}.
		\end{equation}
		Therefore, Eq.~\eqref{applog2} can be simplified to
		
		\begin{equation}\label{eq:applog3} 
			\begin{aligned}
				\ket{G}=\sum_{i_1,\dots,i_k\in \mathbb{Z}_2}   \ket{i_1 \dots i_k}_\MM \ot \ket{G_{i_1\dots i_k}}_\PP
				\quad \text{with} \quad  
				\ket{G_{i_1\dots i_k}}:= (-1)^{\sum_{(u,v)\in \E_\MM} i_vi_u} \prod^k_{r=1} \Xbar^{i_r}_r \ket{G_{0\dots0}}\,.
			\end{aligned}
		\end{equation}
		We show now that the elements $\ket{G_{i_1 \dots i_k}}$ form an orthonormal basis, $\{\Xbar_r\}^k_{r=1}$ act as the logical $\Xbar$~gates and $\ket{G_{0\dots0}}$ as the logical zero state. The orthogonality of
		the basis elements $\ket{G_{i_1\ldots i_k}}$ can be seen from 
		\begin{align}
			\braket{G_{j_1\cdots j_k}}{G_{i_1\cdots i_k}} 
			&\propto 
			\Big(
			\bra{+}^{\ot n}\prod_{(u,v)\in \E_\PP}  \CZ_{uv}  \prod^k_{r=1} \Xbar^{-j_r}_r 
			\Big)
			\Big(
			\prod^k_{r'=1}
			\Xbar^{i_{r'}}_{r'} \prod_{(u',v')\in \E_\PP}  \CZ_{u'v'}  \ket{+}^{\ot n} 
			\Big)
			\nonumber\\
			&=
			\bra{+}^{\ot n} \prod^k_{r=1} \Xbar^{-j_r}_r
			\prod^k_{r'=1} \Xbar^{i_{r'}}_{r'}
			\ket{+}^{\ot n} \\
			&=\bra{+}^{\ot n} \prod^k_{r=1} \Xbar^{(i_r - j_r)}_r
			\ket{+}^{\ot n} \,,\label{eq:orth}
		\end{align}
		where we used that the logical $\Xbar$~gates are composed of $\pZ$ gates, and therefore, they commute with $\CZ$ gates and cancel each other.
		Recall that $\pZ\ket{+}=\ket{-}$ so that, Eq.~\eqref{eq:orth} is different to $0$ if and only if $j_r=i_r$, for all values of $r$. Consequently the elements $\ket{G_{i_1\cdots i_k}}$ form a computational basis of a subspace which is identified as code space.
		
		The expression in Eq.~\eqref{eq:applog3} matches that of Eq.~\eqref{eq:graphst_app}. This ends the proof.
	\end{proof}
	
	\section{Logical gates and code subspace}\label{app:OpCode}
	
	\begin{restatable}{fact}{factlog}\label{factlog} Let $\ket{G}$ be the graph state given by   Eq.~\eqref{eq:graphst}. There are elementary row operations on the boundary part of its check matrix, written in Eq.~\eqref{eq:checkphy}, 
		such that:
		\begin{itemize}[noitemsep,topsep=0pt,parsep=0pt,partopsep=0pt]
			\item[i)] the first $n-k$ 
			rows contain the code generators,
			\item[ii)]
			the next $k$ rows contain the logical $\Zbar$~gates,
			\item[iii)] 
			the last $k$ rows contain the logical $\Xbar$~gates.
		\end{itemize}
		Performing such elementary row operations on Eq.~\eqref{eq:checkphy} leads to   
		\begin{equation}\label{eq:checkLogOpxz_app}
			\begin{pmatrix}[c]
				\CC_{G} \\ 
				\thickbar{Z}_{}  \\
				\thickbar{X}_{} 
			\end{pmatrix}=
			\begin{pmatrix}[cc|cc]
				\one_{n-k} & -B^\intercal_1(B_2^\intercal)^{-1}& [\Gamma_{1}-B^\intercal_1(B_2^\intercal)^{-1}\Gamma_2] &  [\Gamma^\intercal_{2}-B^\intercal_1(B_2^\intercal)^{-1}\Gamma_3] \\
				0 & (B^\intercal_2)^{-1}& (B^\intercal_2)^{-1}\Gamma_2 & (B^\intercal_2)^{-1}\Gamma_3 \\ 
				0 & 0 & B_1 & B_2 \\ 
			\end{pmatrix}\,.
		\end{equation}
	\end{restatable}
	\begin{proof}
		
		We block the check matrix of Eq.~\eqref{eq:checkphy} in the following way,
		\begin{equation}
			\begin{array}{c}
				\text{\scriptsize $R_0$} \\
				\text{\scriptsize $R_1$}\\
				\text{\scriptsize $R_2$}
			\end{array}
			\left(
			\begin{array}{cc|cc}
				\one_{n-k} & 0  &  \Gamma_1 &  \Gamma^\intercal_2  \\
				0 &  \one_k & \Gamma_2 & \Gamma_3 \\
				0 & 0 & B_1 &  B_2
			\end{array}
			\right)\,.
		\end{equation}
		Recall that the code generators and logical gates must satisfy the following relations:
		\begin{align}
			[g_i, g_j] &= 0 
			\label{eq:com0} \\ 
			[g_i, \Xbar_j] &=[g_i,\Zbar_j]=0
			\label{eq:com1} \\
			[\Xbar_i, \Xbar _j] &= 
			[\Zbar_i, \Zbar_j] = 0 
			\label{eq:com2} \\
			[\Xbar_i, \Zbar_j]  &=0 
			&\text{for } i \neq j
			\label{eq:com3} \\
			\Xbar_i \Zbar_j  &= - \Xbar_i \Zbar_j 
			&\text{for } i = j \label{eq:com4}
		\end{align}
		We apply elementary row operations to $R_0$, $R_1$ and $R_2$ so that they fulfil the above commutations relations, leading to $\CC_G$, $\thickbar{Z}$ and $\thickbar{X}$, respectively.
		
		We obtain $\CC_G$ by applying $R_{0} \xrightarrow{\text{row oper.}} \CC_{G}=R_{0}-{B^\intercal _1}(B^\intercal_2)^{-1}R_1$. Hence, 
		\begin{equation}\label{eq:demo}
			\begin{array}{c}
				\text{\scriptsize $\CC_{G}$} \\
				\text{\scriptsize $R_1$}\\
				\text{\scriptsize $R_2$}
			\end{array}
			\left(
			\begin{array}{cc|cc}
				\one_{n-k} & -B^\intercal_1(B_2^\intercal)^{-1} & [\Gamma_{1}-  B^\intercal_1(B_2^\intercal)^{-1}\Gamma_2] & [\Gamma^\intercal_{2}-B^\intercal_1(B_2^\intercal)^{-1}\Gamma_3]  \\ 
				0 &  \one_k & \Gamma_2 & \Gamma_3 \\ 
				0 & 0 & B_1 & B_2 
			\end{array}
			\right)\,.
		\end{equation}
		
		We obtain $\thickbar{Z}_{}$ by applying $R_{1} \xrightarrow{\text{row oper.}} \thickbar{Z}_{}=(B^\intercal_2)^{-1}R_{1}$ and obtain $\thickbar{X}_{}$ by not performing any row operation to $R_2$. Hence, 
		\begin{equation}
			\begin{array}{c}
				\text{\scriptsize $\CC_{G}$} \\
				\text{\scriptsize $\thickbar{Z}_{}$}\\
				\text{\scriptsize $\thickbar{X}_{}$}
			\end{array}
			\left(
			\begin{array}{cc|cc}
				\one_{n-k} & -B^\intercal_1(B_2^\intercal)^{-1} & [\Gamma_{1}-  B^\intercal_1(B_2^\intercal)^{-1}\Gamma_2] & [\Gamma^\intercal_{2}-B^\intercal_1(B_2^\intercal)^{-1}\Gamma_3] \\
				0 & (B^\intercal_2)^{-1} & (B^\intercal_2)^{-1}\Gamma_2 & (B^\intercal_2)^{-1}\Gamma_3 \\
				0 & 0 &  B_1 & B_2
			\end{array}
			\right)\,.
		\end{equation}
		
		With the commutation relations of Eq.~\eqref{eq:pc_comm}, it can be checked that $\CC_G$, $\thickbar{Z}$ and $\thickbar{X}$ satisfies the relations Eq.~\eqref{eq:com0} to Eq.~\eqref{eq:com4}.
		
		Now we show that $\CC_G$, $\thickbar{Z}$ and $\thickbar{X}$ describes the code corresponding to the state Eq.~\eqref{eq:graphst}. 
		First, we note that $\thickbar{X}_{}$ corresponds to the logical $\Xbar$~gates described in Eq.~\eqref{eq:logical0X}. 
		This can be seen by noting that $\thickbar{X}_{}$ contains only non-trivial elements in the $\pZ$-part of the check matrix. 
		Therefore, $\thickbar{X}_{}$ can be written in terms of Pauli matrices as
		\begin{equation}\label{eq:rowtopaulis}
			R_2=\begin{pmatrix}[c|c]
				0 & B
			\end{pmatrix} \xrightarrow[\text{Pauli matrices}]{\text{Check matrix to}} \begin{bmatrix}
				\pZ^{B_{11}} & \ot & \pZ^{B_{k1}} & \ot & \cdots & \ot & \pZ^{B_{1n}} \\
				\vdots &  \ot &  \vdots &  \ot &  \vdots&   \ot&   \vdots \\
				\pZ^{B_{k1}}&  \ot & \pZ^{B_{k2}} & \ot & \cdots&  \ot&  \pZ^{B_{kn}}
			\end{bmatrix}\,. 
		\end{equation}
		The $r$-th row of the right-hand side of 
		Eq.~\eqref{eq:rowtopaulis} is $\Xbar_r$ of Eq.~\eqref{eq:logical0X}.
		
		Furthermore, the rows from $\CC_G$ and $\bar{Z}$ are the generators of the logical zero described in Eq.~\eqref{eq:logical0X}, since both are obtained by only involving rows in $R_0$ and $R_1$.
		Thus the rows $R_1,R_2, R_3$ can be identified with code generators and logical gates corresponding to  Eq.~\eqref{eq:check0}. 
		This ends the proof.
		
	\end{proof}
	
	We note that $\Xbar$ and $\Zbar$~gates are not unique since they can always be multiplied by the stabilizers of the code subspace $\CC$, which will not change their action on the code subspace.
	The phases corresponding to the row operations
	in the proof of Fact~\ref{factlog} are:
	\begin{equation}\label{eq:CandZphases}
		(\thickbar{Z}|\omega^\intercal_{\thickbar{Z}})_i=\sum^k_{l=0} [(B^\intercal_2)^{-1}]_{il} (R_1|0^\intercal)_l
		\,\,, \quad
		(\CC_G|\omega^\intercal_\CC)_i=(R_0)_i+\sum^k_{l=0} [B^\intercal_1(B^\intercal_2)^{-1}]_{il} (R_1|0^\intercal)_l \,.
	\end{equation}

	\section{Encoding}\label{app:Encoding}
	
	\begin{restatable}{fact}{factgates}\label{factgates}
		Let $\ket{G}$ be a graph state described by Eq.~\eqref{eq:check0} with logical gates shown in Eq.~\eqref{eq:checkLogOpxz}. Then, a bulk state $\ket{\phi_{\inc}}$ encodes into a boundary state $\ket{\phi_{\enc}}$ via
		\begin{equation}
			\label{eq:gateUG_app}
			\ket{+}^{\otimes k}_{\MM} \ot \ket{\phi_{\enc}}_{\PP} =
			U_{\op{G}}
			\big(
			\ket{\phi_{\inc}}_{\MM} \ot \ket{+}_{\PP}^{\otimes n}  
			\big)\,,
		\end{equation}
		where $U_{\op{G}} = U_3 U_2 U_1$ with
		\begin{equation}\label{eq:gates_app}
			\begin{aligned}
				U_1  & = \!\!\!\!\!\!\!\! \prod_{
					\substack{
						(u,v)  \in \{\E_\MM, \E_\PP\} 
					} 
				} \!\!\!\!\!\! \CZ_{uv}\,, \quad U_2 =\prod^{k}_{j=1} \CXbar_{j}\,, \\
				U_3 &= \prod^{k}_{j=1} {\CZbar}_{j} 
				\Big( 
				\Had^{\otimes k}_{\MM} \ot \one^{\ot n}_{\PP} 
				\Big) \,.
			\end{aligned}
		\end{equation}
	\end{restatable}
	
	\begin{proof}
		Given $\ket{G}$ described in Eq.~\eqref{eq:graphst}, the isometry $T_{\op{G}}$ can be written as
		\begin{equation}\label{eq:TG}
			T_{\op{G}}= \sum_{\bmR{i}\in \mathbb{Z}^k_2} \ket{G_{\bmR{i}}}_\PP  \bra{\bmR{i}}_\MM\;,
		\end{equation}
		with $\bmR{i}=(i_1,\dots,i_k) \in \mathbb{Z}^k_2$. 
		Given an initial bulk state as $\ket{\phi_{\inc}}_{\MM}=\sum_{\bmR{i}\in\mathbb{Z}^k_2} c_{\bmR{i}} \ket{\bmR{i}}_{\MM}$, the encoded state should then
		read
		\begin{equation}\label{eq:encodedstate}
			\ket{\phi_{\enc}}_{\PP}=T_{\op{G}}\ket{\phi_{\inc}}_{\MM} =\sum_{{\bmR{i}}\in\mathbb{Z}^k_2} c_{\bmR{i}} \ket{G_{\bmR{i}}}_\PP\,.
		\end{equation}
		
		To obtain the encoding through Eq.~\eqref{eq:gateUG_app}, 
		initialize the boundary state in 
		$\ket{+}_{\PP}^{\otimes n}$ 
		and apply $U_1$,
		\begin{align}
			U_1 \big(\ket{\phi_{\inc}}_\MM \ot \ket{+}_{\PP}^{\otimes n} \big) &= \!\!\!\!\!\!\!\! \prod_{
				\substack{
					(u,v)  \in \{\E_\MM, \E_\PP\}  
				} 
			} \!\!\!\!\!\! \CZ_{uv}  \big(\ket{\phi_{\inc}}_\MM \ot \ket{+}_{\PP}^{\otimes n}\big) \nonumber \\ 
			& = \prod_{(u,v)\in \E_\MM} {\CZ_{uv}}\ket{\phi_{\inc}}_{\MM} \;\; \ot \!  \prod_{(u,v)\in \E_\PP} \CZ_{uv} \ket{+}_{\PP}^{\otimes n} \nonumber \\ 
			&= \prod_{(u,v)\in \E_\MM} {\CZ_{uv}} \ket{\phi_{\inc}}_{\MM} \ot \ket{G_{0\dots0}}_{\PP} \nonumber \\ 
			&= \sum_{{\bmR{i}}\in\mathbb{Z}^k_2} c_{\bmR{i}}(-1)^{\sum_{(u,v)\in \E_\MM}   i_vi_u} \ket{\bmR{i}}_{\MM} \ot \ket{G_{0\dots0}}_{\PP}\label{eq:U1result}\,.
		\end{align}
		Here we used the definition of the logical zero state 
		[see Eq.~\eqref{eq:logical0X}] and added the corresponding real phases as shown in Eq.~\eqref{eq:phases}.
		
		Second, apply $U_2$ to Eq.~\eqref{eq:U1result} and obtain

		\begin{align}
			U_2\Big( \sum_{{\bmR{i}}\in\mathbb{Z}^k_2} c_{\bmR{i}}(-1)^{\sum_{(u,v)\in \E_\MM}   i_vi_u} \ket{\bmR{i}}_{\MM} \ot \ket{G_{0\dots0}}_{\PP} \Big)
			&=\prod^k_{j=1} \CXbar_j \Big( \sum_{{\bmR{i}}\in\mathbb{Z}^k_2} c_{\bmR{i}}(-1)^{\sum_{(u,v)\in \E_\MM}   i_vi_u} \ket{\bmR{i}}_{\MM} \ot \ket{G_{0\dots0}}_{\PP} \Big) \nonumber\\ 
			&=\sum_{{\bmR{i}}\in\mathbb{Z}^k_2}  \ket{\bmR{i}}_{\MM} \ot \Big(c_{\bmR{i}} (-1)^{\sum_{(u,v)\in \E_\MM} i_vi_u}  \prod^k_{r=1}\pX_r^{i_r}\ket{G_{0\dots0}}_{\PP} \Big) \nonumber \\
			&= \sum_{{\bmR{i}}\in\mathbb{Z}^k_2} c_{\bmR{i}}\ket{\bmR{i}}_{\MM}\ot \ket{G_{\bmR{i}}}_{\PP}  \label{eq:U2result}\;,
		\end{align}
		where we use the definition of the logical states [Eq.~\eqref{eq:logicalbasis}].
		
		Finally, we apply $U_3$ to Eq.~\eqref{eq:U2result} and obtain
		\begin{align}
			U_3 \Big(\sum_{{\bmR{i}}\in\mathbb{Z}^k_2} c_{\bmR{i}}\ket{\bmR{i}}_{\MM} \ot \ket{G_{\bmR{i}}}_{\PP} \Big)
			&=\prod^k_{j=1} \CZbar_j \Big(  \Had_{\MM}^{\otimes k} \ot \one_{\PP}^{\ot n}  \Big)\sum_{{\bmR{i}}\in\mathbb{Z}^k_2} c_{\bmR{i}}  \ket{\bmR{i}}_{\MM} \ot \ket{G_{\bmR{i}}}_{\PP}  \nonumber\\ 
			&=  \frac{1}{\sqrt{2^{k}}} \prod^k_{j=1} \CZbar_j \sum_{{\bmR{i},\bmR{q}}\in\mathbb{Z}^k_2}  c_{\bmR{i}} (-1)^{\bmR{i} \cdot \bmR{q}}  \ket{\bmR{q}}_{\MM}  \ot \ket{G_{\bmR{i}}}_{\PP} \nonumber\\
			&= \frac{1}{\sqrt{2^{k}}} \sum_{\bmR{q}\in\mathbb{Z}^k_2} \ket{\bmR{q}}_{\MM}\ot \Big( \sum_{{\bmR{i}}\in\mathbb{Z}^k_2}  c_{\bmR{i}} \ket{G_{\bmR{i}}}_{\PP} \Big) \,. \label{eq:appgates1}
		\end{align}
		Here we used that $\Had\ket{i}=\frac{1}{\sqrt{2}}\sum_{q\in\mathbb{Z}_2} (-1)^{i \cdot q} \ket{q}$ and with 
		$\CZbar_j =
		\ketbra{0}{0}_{j} \otimes \one +\ketbra{1}{1}_{j}  \otimes \Zbar$
		that
		\begin{equation}
			\prod^k_{j=1} \CZbar_j \Big( \ket{\bmR{q}}_{\MM}\ot \ket{G_{\bmR{i}}}_{\PP} \Big)=
			\prod^k_{r=1}  \ket{\bmR{q}}_{\MM} 
			\ot
			\Zbar^{q_r}_r 
			\ket{G_{\bmR{i}}}_{\PP}  =
			\prod^k_{r=1} (-1)^{i_r\cdot q_r} \ket{\bmR{q}}_{\MM}\ot \ket{{G_{\bmR{i}}}}_{\PP}= 
			(-1)^{\bmR{i} \cdot \bmR{q}} \ket{\bmR{q}} \ot\ket{G_{\bmR{i}}}_{\PP}\;.
		\end{equation}
		After applying those gates, Eq.~\eqref{eq:appgates1} becomes a tensor product state
		\begin{equation}
			\begin{aligned}
				\frac{1}{\sqrt{2^{k}}} \sum_{\bmR{q}\in\mathbb{Z}^k_2} \ket{\bmR{q}}_{\MM} \ot \Big( \sum_{{\bmR{i}}\in\mathbb{Z}^k_2}  c_{\bmR{i}} \ket{G_{\bmR{i}}}_{\PP} \Big) = 
				\ket{+}_{\MM}^{\otimes k} \ot     \ket{\phi_{\enc}}_{\PP} \;,
			\end{aligned}
		\end{equation}
		where we use the definition of encoded state [see Eq.~\eqref{eq:encodedstate}]. Hence,
		\begin{equation}
			U_3 U_2 U_1 \big( \ket{\phi_{\inc}}_{\MM} \ot \ket{+}_{\PP}^{\otimes n} \big)=  \ket{+}^{\otimes k}_{\MM} \ot \ket{\phi_{\enc}}_{\PP}\,.
		\end{equation}
		This ends the proof.
	\end{proof}
	
	\section{Consequences of local Clifford operations}\label{app:Pdec}
	
	Contracting the tensor network of the hyperbolic pentagon code leads to the holographic state $\ket{H}$, 
	which represents an isometry $T_{\op{H}}$ that encodes $k$ bulk qubits into $n$ boundary qubits.
	Performing this contraction only partially, so that it involves the bulk region $\regE$ and its nearby boundary $\regdE$, leads to $\ket{h}$. 
	This state can then be understood as an isometry $T_h$ from $\regE \cup \gamma$ to $\regdE$. 
	Given a state which was encoded through the isometry $T_{\op{H}}$, one can apply~$T^{\dagger}_{\op{h}}$ to $\regdE$ and recover the bulk region $\regE$. 
	
	How is this partial recovery performed for the local Clifford 
	optimized state
	$\ket{G} = \Vop \ket{H}$ [see Eq.~\eqref{eq:Holotograph}]? 
	With Eq.~\eqref{eq:state} write 
	\begin{equation}\label{eq:HtoV}
		\ket{G}=(\Vop_{\MM} \ot \Vop_{\PP})\ket{H} =\!\!\sum^1_{i_1, \dots, i_k=0}    \Vop_{\MM} \ket{i_1 \dots i_k}_\MM \ot \Vop_{\PP} \ket{H_{i_1 \dots i_k}}_\PP \,,
	\end{equation}
	where ${V_\MM}$ and ${V_{\PP}}$ 
	contain the local Cliffords of $\Vop$ having support on $\MM$ and $\PP$, respectively.
	Flipping the ket acting on $\MM$ into a bra
	then yields the corresponding isometry
	\begin{equation}\label{eq:AppPdec0}
		T_{\op{G}}
		=
		\Vop_{\PP} T_H \Vop_{\MM}
		=
		\sum^1_{i_1, \dots, i_k=0}   \Vop_{\PP} 
		\ket{H_{i_1 \dots i_k}}_\PP 
		\bra{i_1 \dots i_k}_\MM \Vop_{\MM}\,.
	\end{equation}
	This establishes a relation between the encoding via $\ket{H}$ and that via $\ket{G}$, and in turn, also for the partial isometries 
	$\ket{h}$ and $\ket{g}$ as we explain now.
	While for the encoding $T_{\op{H}}$ a partial decoding is given by $T^\dagger_{\op{h}}$, Eq.~\eqref{eq:AppPdec0} shows that the encoding $T_{\op{G}}$ is partially decoded by 
	\begin{equation}\label{eq:TGh}
		\widetilde{T}^\dagger_{\op{h}}=(\Vop^\dagger_{\regE} \ot \one_\gamma)T^\dagger_{\op{h}} \Vop^\dagger_{\regdE} \,.
	\end{equation}
	Here $\Vop_\regE$ and $\Vop_{\regdE}$ 
	contain the local Cliffords of $\Vop$ having support on $\regE$ and $\regdE$.
	
	This isometry $\tilde{T}_{\op{h}}$ partially decodes a state encoded through $T_{\op{G}}$. 
	We now need to derive the corresponding quantum gates.
	For this,  consider the local Clifford map from $\ket{h}$ to $\ket{g}$.  
	By applying the encoding method described in Eq.~\eqref{eq:gateUG} for $\ket{g}$, we obtain the encoding unitary ${U_g}$ that is equivalent to $T_{\op{g}}$.
	Analogous to Eq.~\eqref{eq:AppPdec0}, the isometries $T_{\op{g}}$ and $T_{\op{h}}$ are related by 
	\begin{equation}\label{eq:th}
		T^\dagger_{\op{h}}=(\Wop_{\regE}\ot \one_{\gamma}) T^\dagger_{\op{g}} \Wop_{\regdE}\,,
	\end{equation}
	where $\ket{g}=\Wop\ket{h}$ with $\Wop$ a local Clifford gate. Here ${W_{\regE}}$ and ${W_{\regdE}}$ 
	contain the local Cliffords of $\Wop$ having support on $\regE$ and $\regdE$, respectively. Since we are only interested in recovering $\regE$, one can ignore the action of $\Wop$ on $\gamma$. Combining Eqs.~\eqref{eq:TGh} and \eqref{eq:th}, we obtain
	\begin{equation}\label{eq:th1}
		\widetilde{T}^\dagger_{\op{h}}=(\Vop^\dagger_{\regE}\Wop_{\regE} \ot \one_\gamma) T^\dagger_{\op{g}} \Wop_{\regdE}\Vop^\dagger_{\regdE}\,.
	\end{equation}
	
	Eq.~\eqref{eq:th1} can be understood schematically as a chain of isometries
	\begin{equation}
		T^\dagger_{\op{g}} \; \xrightarrow[]{W} \; T^\dagger_{\op{h}} \; \xrightarrow[]{V^\dagger} \; \widetilde{T}^\dagger_{\op{h}} \,.
	\end{equation}
	Finally, we express the isometry $\widetilde{T}_{\op{h}}^\dagger$ as a unitary gate ${\widetilde{U}^\dagger_{\op{h}}}$,
	
	\begin{equation}\label{eq:recgates}
		\widetilde{U}^\dagger_{\op{h}}=(\Vop^\dagger_{\regE}\Wop_{\regE} \ot \one_{\gamma\regdE})\,U^\dagger_{\op{g}}\, (\one_{\regE\gamma}\ot \Wop_{\regdE}\Vop^\dagger_{\regdE})\,.
	\end{equation}
	For the  graph $\ket{G^{\text{opt}}}$ illustrated in Fig.~\ref{fig:MinimalGraph}, the unitary $\Vop = {\Vop}_{\PP} \ot {\Vop}_{\MM}$ is
	\begin{equation}\label{eq:unitaryV}
		\begin{matrix}[c ccc ccc ccc ccc c cccc]
			&\text{\footnotesize1}&\text{\footnotesize2}&\text{\footnotesize3}&
			\text{\footnotesize4}&\text{\footnotesize5}&\text{\footnotesize6}&
			\text{\footnotesize7}&\text{\footnotesize8}&\text{\footnotesize9}&
			\text{\footnotesize10}&\text{\footnotesize11}&\!\!\!\text{\footnotesize12} 
			
			& & \text{\footnotesize A}&\text{\footnotesize B}&\text{\footnotesize C}&\!\!
			\!\text{\footnotesize D}\\
			{{{\Vop}_{\PP}}} =&
			\one&\one&\Had&
			\one&\one&\Had&
			\one&\one&\Had&
			\one&\one&\Had \,\,,
			&\quad{{{\Vop}_{\MM}}}  =&\pZ&\pZ&\pZ&\pZ \,.
		\end{matrix}
	\end{equation}
	\begin{figure}
		\centering    
		\includegraphics[width=0.8\textwidth]{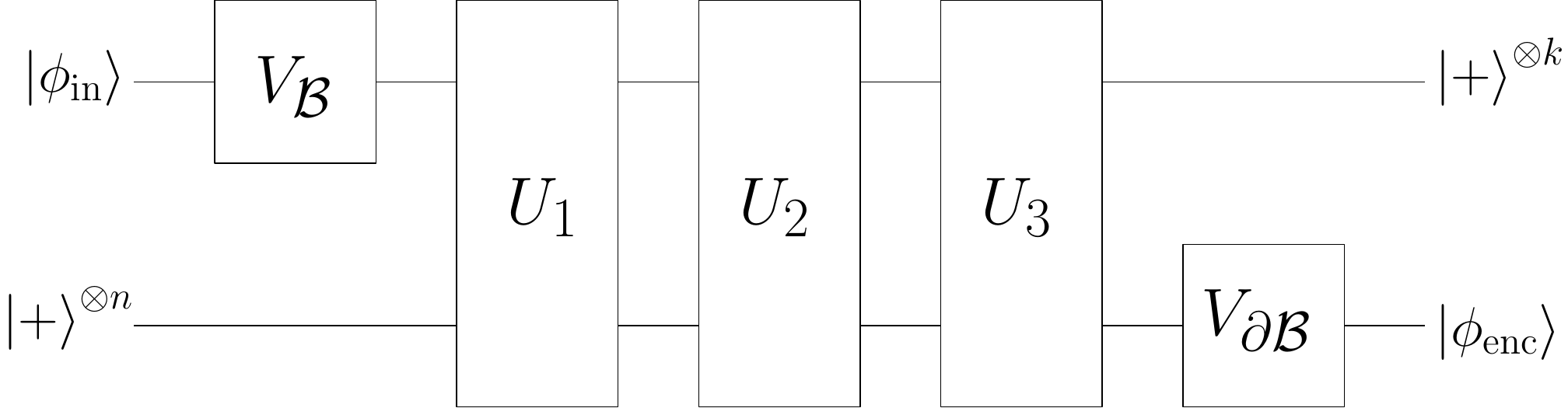}    \caption{\textbf{Local Clifford-equivalent encoding.} 
			Gates used to correct the encoding, when one wants to use the exact holographic state for the encoding $\ket{H}$ instead of its LC-equivalent state $\ket{G}$. The unitary gates ${U_1}$, ${U_2}$ and ${U_3}$ are the encoding gates described in Eq.~\eqref{eq:gates}.}
		\label{fig:gates}
	\end{figure}For the graph  $\ket{g}$ illustrated in Fig.~\ref{fig:rec}, the unitary $\Wop = {\Wop}_{\regdE} \ot {\Wop}_{\regE\gamma}$ is
	\begin{equation}\label{eq:unitaryW}
		\begin{matrix}[cccccc cccccc]
			&\text{\footnotesize1}&\text{\footnotesize2}&\text{\footnotesize3}&\text{\footnotesize4}&\!\!
			\!\text{\footnotesize5}
			& &\text{\footnotesize \text{I}}&\text{\footnotesize \text{A}}&{\text{\footnotesize III}}&{\text{\footnotesize II}}&\!\!
			\!\text{\footnotesize B} \\
			{{{\Wop}_{\regdE}}} =&
			\one&\one&\Had&\one&\one\,,& 
			\quad {{{\Wop}_{\regE\gamma}}}  =&\one&\pZ&\one&\one&\one\,. 
		\end{matrix}
	\end{equation}
	
	Note that for these two graphs ($\ket{G^{\text{opt}}}$ and $\ket{g}$), Eq.~\eqref{eq:recgates} is reduced to    $\widetilde{U}^\dagger_{\op{h}}=Z_BU^\dagger_{\op{g}}$.
	\section{Optimal code subspace generators, logical gates and fidelity calculation}\label{app:fidelity}
	
	From Eq.~\eqref{eq:checkLogOpxz}, the subspace generators of the graph code illustrated in Fig~\ref{fig:MinimalGraph} can be obtained:
	\begin{equation}\label{eq:gen}
		\begin{matrix}[ccc ccc ccc ccc ccc]
			&&&\text{\footnotesize1}&\text{\footnotesize2}&\text{\footnotesize3}&\;
			\text{\footnotesize4}&\text{\footnotesize5}&\text{\footnotesize6}&\;
			\text{\footnotesize7}&\text{\footnotesize8}&\text{\footnotesize9}&\; \text{\footnotesize10}&\!\text{\footnotesize11}&\!\!\!\!\text{\footnotesize12}\\
			g_1 &=&&
			\pX&\pZ&\pX&\, 
			\pZ&\pZ&\pZ&\,
			\one&\one&\one&\,
			\one\!&\one&\!\one \,, \\
			g_2 &=&&
			\one&\one&\one&\,
			\pX&\pZ&\pX&\, 
			\pZ&\pZ&\pZ&\,
			\one\!&\one&\!\one \,, \\
			g_3 &=&&
			\one&\one&\one&\,
			\one&\one&\one&\,
			\pX&\pZ&\pX&\, 
			\pZ\!&\pZ&\!\pZ\, \,, \\
			g_4 &=&&
			\pZ&\pZ&\pZ&\,
			\one&\one&\one&\,
			\one&\one&\one&\,
			\pX\!&\pZ&\!\pX\,, \\ 
		\end{matrix} \quad
		\begin{matrix}[ccc ccc ccc ccc ccc]
			&&&\text{\footnotesize1}&\text{\footnotesize2}&\text{\footnotesize3}&\;
			\text{\footnotesize4}&\text{\footnotesize5}&\text{\footnotesize6}&\;
			\text{\footnotesize7}&\text{\footnotesize8}&\text{\footnotesize9}&\; \text{\footnotesize10}&\!\text{\footnotesize11}&\!\!\!\!\text{\footnotesize12}\\
			g_5 &=&&
			\pY&\pX&\pZ&\, 
			\pX&\pX&\pY&\,
			\one&\one&\one&\,
			\one\!&\one&\!\one \,, \\
			g_6 &=&&
			\one&\one&\one&\,
			\pY&\pX&\pZ&\, 
			\pX&\pX&\pY&\,
			\one\!&\one&\!\one \,, \\
			g_7 &=&&
			\one&\one&\one&\,
			\one&\one&\one&\,
			\pY&\pX&\pZ&\, 
			\pX \!&\pX& \! \pY\,, \\
			g_8 &=&& 
			\pX&\pX&\pY&\,
			\one&\one&\one&\,
			\one&\one&\one&\,
			\pY\!&\pX&\!\pZ \,. \\  
		\end{matrix}
	\end{equation}
	Recall that the code subspace generators are not unique, since multiplication between them can also be used as a generator. Eq.~\eqref{eq:gen} shows those that minimize the number of non-trivial Pauli elements. 
	Similarly, one can reduce the number of non-trivial elements from the logical gates [Eq.~\eqref{eq:LogicalX} and \eqref{eq:LogicalZ}] by multiplying said gates by code subspace generators. This procedure leads to
	\begin{equation}\label{eq:LogicalXZopt}
		\begin{matrix}[cc ccc ccc ccc ccc]
			&&\text{\footnotesize1}&\text{\footnotesize2}&\text{\footnotesize3}&\;
			\text{\footnotesize4}&\text{\footnotesize5}&\text{\footnotesize6}&\;
			\text{\footnotesize7}&\text{\footnotesize8}&\text{\footnotesize9}&\; \text{\footnotesize10}&\!\text{\footnotesize11}&\!\!\!\!\text{\footnotesize12}\\
			\thickbar{Z}_{\op{A}} &=&
			\pZ&\pX&\pX&\;
			\one&\one&\one&\;
			\one&\one&\one&\;
			\one\!&\one&\!\one \,,\\
			\thickbar{Z}_{\op{B}}&=&
			\one&\one&\one&\;
			\pZ&\pX&\pX&\;
			\one&\one&\one&\; 
			\one&\one&\one\,,\\
			\thickbar{Z}_{\op{C}}&=&
			\one&\one&\one&\;
			\one&\one&\one&\;
			\pZ&\pX&\pX&\;
			\one&\one&\one\,,\\ 
			\thickbar{Z}_{\op{D}}&=&
			\one&\one&\one&\;
			\one&\one&\one&\;
			\one&\one&\one&\;
			\pZ&\pX&\pX\,,
		\end{matrix}\quad
		\begin{matrix}[ccc ccc ccc ccc ccc]
			&&&\text{\footnotesize1}&\text{\footnotesize2}&\text{\footnotesize3}&\;
			\text{\footnotesize4}&\text{\footnotesize5}&\text{\footnotesize6}&\;
			\text{\footnotesize7}&\text{\footnotesize8}&\text{\footnotesize9}&\; \text{\footnotesize10}&\!\text{\footnotesize11}&\!\!\!\!\text{\footnotesize12}\\
			\thickbar{X}_{\op{A}} &=&-& \pY &\pZ& \pY&\,
			\one& \one& \one&\, 
			\one&\one&\one&\,
			\one&\one&\one \,, \\
			\thickbar{X}_{\op{B}}&=&-&
			\one&\one&\one&\,
			\pY &\pZ& \pY&\,
			\one& \one& \one& 
			\one&\one&\one \,, \\
			\thickbar{X}_{\op{C}} &=&-&
			\one&\one&\one&\,
			\one&\one&\one&\,
			\pY &\pZ& \pY&\,
			\one &\one&\one \,, \\
			\thickbar{X}_{\op{D}} &=&-&
			\one&\one&\one&\,
			\one&\one&\one&\,
			\one & \one & \one &
			\pY&\pZ&\pY \,.
		\end{matrix}
	\end{equation}
	With this optimization the number of gates is reduced to perform the encoding shown in \hyperref[sect:Encircuit]{\textit{Encoding circuit}}. Nevertheless, we lose the geometric interpretation of logical $\Xbar$ gates described in Eq.~\eqref{eq:LogicalX}.
	
	The logical zero state is stabilized by $\langle g_1,\dots,g_8, \Zbar_A, \Zbar_B, \Zbar_C, \Zbar_D\rangle$. 
	A way to prepare the logical zero state is through non-destructive
	stabilizer measurements \cite[Chapter 10.5.8]{nielsen_chuang_2010}, as for
	example recently done in a trapped-ion setup~\cite{Abobeih_2022}. The number of two-qubit gates to encode required by this method is the same as the number of non-trivial elements of the generators of the logical zero state plus two single-qubit gates per measurement.
	Multiplying $\{g_i\}$ by the $\Zbar$ logic gates, we can obtain eight new generators $\{g'_i\}$, with only four non-trivial Pauli elements, that substitute $\{g_i\}$, e.g. $g'_1=\Zbar_{\op{A}}, \Zbar_{\op{B}} g_1$.
	Therefore, the stabilizer of the logic zero state is also described by $\langle g'_1,\dots,g'_8, \Zbar_{\op{A}}, \Zbar_{\op{B}}, \Zbar_{\op{C}}, \Zbar_{\op{D}}\rangle$. 
	Taking that into account, we require $44$ two-qubit gates, $24$ single-qubit gates and $12$ measurements to encode the logical zero state via stabilizer measurements. 
	On the other hand, in \hyperref[sect:preparingLog0]{\textit{Preparing the logical states}} we show that the logical zero state can be implemented as a graph state with only $28$ two-qubit gates and $12$ single-qubit gates.
	
	Both methods can be compared by calculating the circuit fidelity which generates the logical zero state. The fidelity of the quantum circuit can be estimated by
	\begin{equation}\label{eq:fidelity}
		F=(F_m)^{N_m} (F_2)^{N_2} (F_1)^{N_1} \quad \text{with} \quad \frac{\delta F}{F}=\sqrt{\Big(N_1\frac{\delta F_1}{F_1}\Big)^{2} + \Big(N_2\frac{\delta F_2}{F_2}\Big)^{2}+\Big(N_m\frac{\delta F_m}{F_m}\Big)^{2} }\,.
	\end{equation}
	Here $N_1$ and $N_2$ are the number of single-qubit and two-qubit gates respectively and $N_m$ is the number of measurements.
	Ref.~\cite{Quantinuum} proposes a trapped-ion setup with single-qubit and two-qubit gate fidelity of $0.99994(3)$ and $0.9981(3)$ respectively. Furthermore, the measurement fidelity of the device is $0.9972(5)$. 
	This setup allows to calculate the fidelity of the logical zero state. By Eq.~\eqref{eq:fidelity}, we estimate a fidelity of $0.947(8)$ using our method and of $0.88(2)$ using stabilizer measurements. The fidelity calculation shows a significant advantage in preparing the logical zero state as a graph state.
	
	\section{Larger instance of the hyperbolic pentagon code}\label{app:larger}
	
	Our work is mainly focused in implementing the hyperbolic pentagon code for the small instance shown in Figure~\ref{fig:holo_larger}a. %
	This corresponds to contracting four AME states, transforming the result to a graph state and reducing the number of edges and its range interaction following the optimization process described in
	\hyperref[sect:holographstate]{\textit{Holographic graph state}}. One optimization method is via the algorithm described in Ref.~\cite{Adcock2020mappinggraphstate}. However, this method does not work for larger instances of the hyperbolic pentagon code. 
	Another method consists in heuristically applying Hadamard gates to transform $\ket{H}$ to $\ket{G}$ as shown in Eq.~\eqref{eq:Holotograph}. This is the method used in this section to find an improved graph state for a larger instance of the hyperbolic pentagon code.
	
	We contract eleven AME states following the tensor network configuration shown in the highlighted region from Figure~\ref{fig:holo_larger}a.
	The result is a stabilizer state of $36$ qubits which can be transformed to the graph state shown in Figure~\ref{fig:holo_larger}b by applying 
	$20$ Hadamard gates to the highlighted qubits in purple. The used instance encodes $11$ bulk to $25$ boundary qubits.
	
	Note that the graph state from Figure~\ref{fig:holo_larger}b is also rotational invariant and the interactions has some locality attached.
	The graph has $195$ edges, $85$ of those correspond to connections between boundary qubits. 
	Therefore, we would need $85$ controlled-$\pZ$ gates to implement the logical zero.

	\begin{figure*}[tbp]
		\centering
		\includegraphics[width=0.8\textwidth]{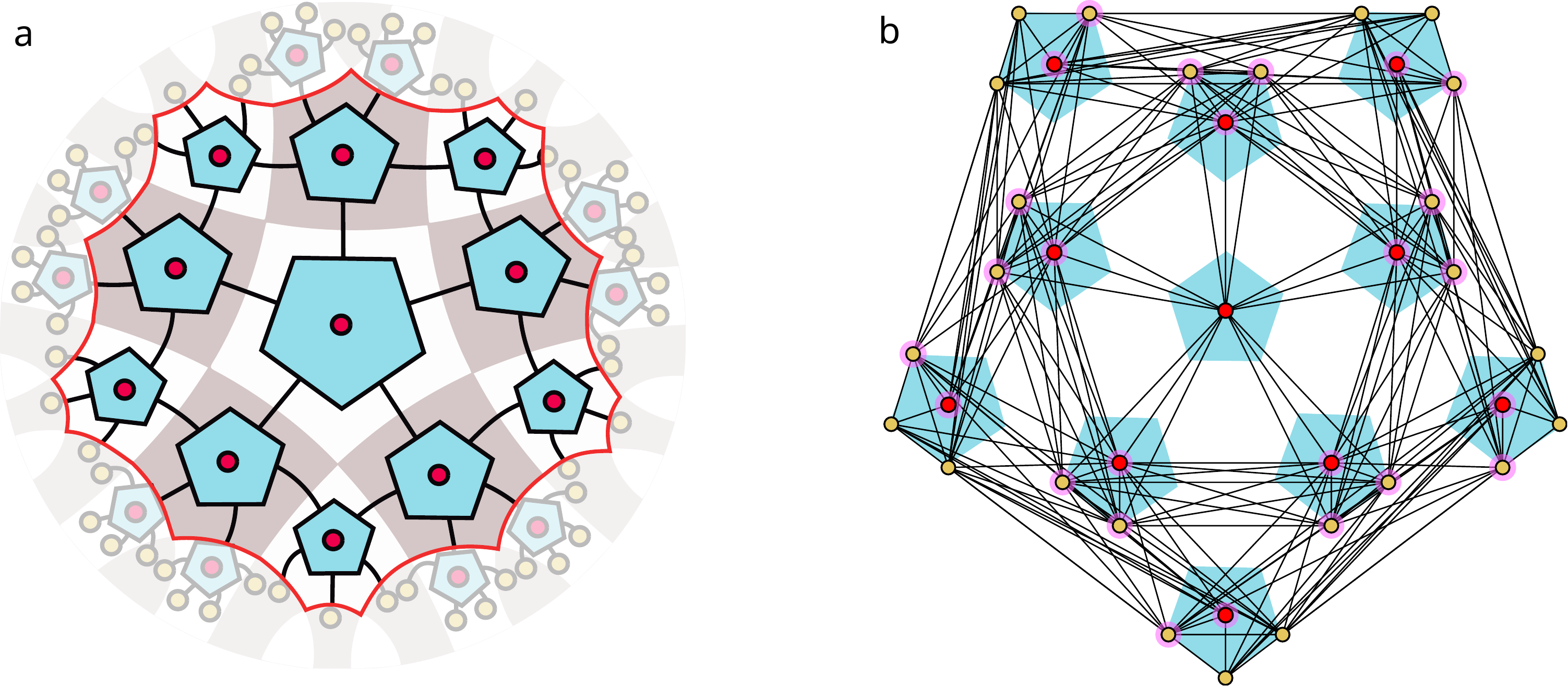}
		\caption{\textbf{Larger instance of the hyperbolic pentagon code transformed to a graph state.} As in Fig.~\ref{fig:holo}b, the left figure shows the highlighted region of the hyperbolic pentagon code which corresponds to the graph code in the right. 
			The AME states (blue pentagons) are contracted following the tensor network representation from Figure~\ref{fig:holo_larger}a leading to the state $\ket{H}$.
			This state encodes $11$ bulk (red) qubits  to $25$ in the boundary (golden) qubits. 
			To transform $\ket{H}$ to the graph state $\ket{G}$ from Figure~\ref{fig:holo_larger}b, we applied $20$ Hadamard gates to the qubits highlighted in purple. The resulting graph state has $195$ edges and $36$ nodes.}
		\label{fig:holo_larger}
	\end{figure*}
	
	\bibliography{Notes}
	\end{document}